\documentclass[11pt, letterpaper]{article}
\usepackage[utf8]{inputenc}
\usepackage{microtype}
\usepackage{mathtools,amssymb,mathrsfs,amsthm}   
\usepackage[table,svgnames]{xcolor}
\usepackage{breakurl}
\usepackage[hyperindex,breaklinks,colorlinks=true,linkcolor=black,citecolor=MidnightBlue,urlcolor=MidnightBlue,hypertexnames=false]{hyperref}
\usepackage{hyperref}
\usepackage{xspace}
\usepackage{graphicx}
\usepackage{tikz}
\usetikzlibrary{chains,arrows,arrows.meta} 
\usepackage[letterpaper,tmargin=1.2in,bmargin=1.2in,hmargin=1.1in]{geometry}
\usepackage{floatpag}
\usepackage{marvosym}
\usepackage{framed}
\usepackage{appendix}

\usepackage{enumitem}
\usepackage{algorithm}
\usepackage{algorithmicx}
\usepackage{algpseudocode}
\usepackage{comment}
\usepackage{verbatim}
\usepackage{cleveref}
\newtheorem{theorem}{Theorem}[section]
\newtheorem{lemma}[theorem]{Lemma}

\newtheorem{corollary}[theorem]{Corollary}

\newtheorem{observation}[theorem]{Observation}

\newtheorem{definition}{Definition}

\usepackage{titlecaps}

\newcommand{\eps}{\varepsilon}
\newcommand{\T}{\mathcal{T}}
\newcommand{\E}{\mathbb{E}}

\newcommand{\Eterm}{\mathcal{E}_{\text{term}}}

\newcommand{\good}{$\pi'$-aligned\xspace}
\newcommand{\Good}{$\pi'$-Aligned\xspace}

\newcommand{\EFNM}{unbounded probabilistic Erase--Flip noise model\xspace}

\author{
  Eden Fargion\\
  \texttt{eden.fargion@gmail.com}\\
  Bar-Ilan University
  \and
  Ran Gelles\\
  \texttt{ran.gelles@biu.ac.il}\\
  Bar-Ilan University
  \and
  Meghal Gupta\\
  \texttt{meghal@berkeley.edu}\\
  UC Berkeley
}
\title{\vspace{-2em}\textbf{Interactive Coding with Unbounded Noise}}
\date{}

\begin{document}

\maketitle

\begin{abstract}
Interactive coding allows two parties to conduct a distributed computation despite noise corrupting a certain fraction of their communication. Dani et al.\@ (Inf.\@ and Comp., 2018) suggested a novel setting in which the amount of noise is unbounded and can significantly exceed the length of the (noise-free) computation. 
While no solution is possible in the worst case, under the restriction of \emph{oblivious} noise, Dani et al.\@ designed a coding scheme that succeeds with a polynomially small failure probability.

We revisit the question of conducting computations under this harsh type of noise and devise a computationally-efficient coding scheme that guarantees the success of the computation, except with an \emph{exponentially} small probability. This higher degree of correctness matches the case of coding schemes with a bounded fraction of noise.

Our simulation of an $N$-bit noise-free computation in the presence of $T$ corruptions, communicates an optimal number of $O(N+T)$ bits and succeeds with probability $1-2^{-\Omega(N)}$.
We design this coding scheme by introducing an intermediary noise model, where an oblivious adversary can choose the locations of corruptions in a worst-case manner, but the effect of each corruption is random: the noise either flips the transmission with some probability or otherwise erases it. This randomized abstraction turns out to be instrumental in achieving an optimal coding scheme.
\end{abstract}

\newpage
\interfootnotelinepenalty=9000
\widowpenalty=8000
\clubpenalty=8000

\section{Introduction}
In many distributed systems nowadays, communication channels suffer various types of noise and interference that may corrupt information exchanged between devices. \emph{Interactive coding}, initiated by the seminal work of Schulman~\cite{schulman92,schulman96} (see also~\cite{gelles17}), allows two or more devices to correctly complete their computation despite channel noise, by adding only a moderate amount of redundancy to the computation. 
The capability of an interactive coding scheme usually depends on the specific type of noise it is designed to withstand. For instance, when the noise can flip bits, interactive coding schemes  withstand up to a fraction of~1/6 of flipped bits~\cite{EGH16,GZ22} (a fraction of~$1/4$ can be withstood over channels with larger alphabets~\cite{BR14}); when the noise erases bits (i.e., replaces a bit with a special erasure mark~$\bot$), then a fraction of~$1/2$ of bit erasures can be withstood~\cite{FGOS15,GZ22}, which also applies for larger alphabets~\cite{EGH16}. When messages can be inserted and deleted, the maximal corruption rate is again~$1/4$, see~\cite{BGMO17,SW19}.

In a recent work, Dani, Hayes, Movahedi, Saia, and Young~\cite{DHMSY18} suggested a different and interesting model for interactive coding in which the amount of noise is \emph{unbounded}. 
That is, the number~$T$ of corruptions that affects a given execution, can be arbitrary.
Note that this number~$T$ is unknown to the coding scheme; this is in contrast to the standard model of interactive coding, where a limit on the fraction of corrupted transmissions is known by all devices.
The scheme in~\cite{DHMSY18} correctly computes any two-party computation that takes $N$~rounds without noise, by communicating $N+ O(T+\sqrt{N(T+1)\log T})$ bits and succeeds with probability~$1-O(1/N\log N)$. 

In a nutshell, the idea of the scheme in~\cite{DHMSY18} is as follows. Every message sent between the parties contains the round number it corresponds to and a signature. A device checks that the signature is valid before processing a received message. If the signature does not check out, the device ignores that communication. 
The coding scheme tracks the progress of both parties via the added information of the round number, so that corrupted messages are re-transmitted until they arrive correctly at the other side. 

One significant drawback of the above approach, is that the noise might corrupt a message along with its signature so that the receiver believes that the signature is correct. This occurs with exponentially small probability in the length of the signature, which leads to the polynomially-small failure probability of the scheme. In other words, the scheme in~\cite{DHMSY18} assumes that the noise \emph{never creates a valid signature} and settles with a failure probability of magnitude~$1/N\log N$.

In this work we aim to achieve an interactive coding scheme that can withstand an unbounded amount of noise, yet, with failure probability exponentially small in~$N$, similar to most previous work on interactive coding (e.g.,~\cite{schulman96,GMS14,BKN14,GH14}). 
This effectively means that the coding scheme must cope with corrupted messages being processed by some device. That is, the coding scheme must be resilient to the event, that occurs with polynomially small probability in~$N$, where both the message and the signature are corrupted in a matching way.

Our main result is a coding scheme that is resilient to an arbitrary and a priori unknown number~$T$ of bit flips, 
with exponentially small failure probability.
\begin{theorem}[Main]\label{thm:main:realModel}
Given any two-party binary interactive protocol $\pi$ of length $N$, 
there exists an efficient randomized protocol $\Pi$ of length 
$O(N+T)$
that simulates $\pi$ 
with probability~$1-2^{-\Omega(N)}$
over a binary channel in the presence of an arbitrary and a priori unknown number~$T$ of corruptions. 
The noise is assumed to be independent of the parties' randomness.
\end{theorem}
We note that the scheme assumes  \emph{oblivious noise} in the sense that the $T$ corrupted transmissions are selected at the beginning of the computation (as a function of the coding scheme and the parties' inputs) and is independent of the parties' (private) randomness. 
This assumption is crucial, as no coding scheme  withstands an unbounded amount of noise that is non-oblivious~\cite[Theorem~6.1]{DHMSY18}.

\subsection{Techniques}
\paragraph{Towards an optimal scheme: the code concatenation approach.}
The immediate approach towards an improved coding scheme for an unbounded amount of corruptions is of \emph{code concatenation}, namely composing two layers of interactive code. The inner layer would be responsible for transmitting bits over the channel despite the unbounded amount of noise (e.g.,~\cite{DHMSY18,GI19}). The outer layer would then ``see" only a limited amount of noise (which passes the inner layer with polynomially-small probability) and perform a standard interactive coding (e.g.,~\cite{schulman92,schulman96,BR14,haeupler14,GH14,GHKRW18,GZ22}) using these bits. 

Unfortunately, such an approach faces severe difficulties. For the inner layer, the scheme of Dani et al.~\cite{DHMSY18} assumes that no corrupted message is accepted by a party. That is, a message can either be correct or marked as invalid. Requiring the parties to process incorrect messages might cause their inner state to differ in a way that could not be recovered  
by the scheme. That is, a single corrupted message (that is believed to be correct by one of the parties) might cause the parties to ``lose sync'', so that the parties do not agree anymore on when the next phase of the scheme begins and ends, or whether the scheme has terminated or not. 
The scheme will not recover from this fault because the synchronization information would be sent by one party at certain rounds but expected by the other party at different rounds.
The other option for the inner layer is the scheme by Gelles and Iyer~\cite{GI19} designed to withstand an unbounded amount of \emph{erasures}, and thus, on its surface, does not fit our purpose.

\paragraph{The randomized erasure--flip model}
The key towards solving the above conundrum stems from defining a new random noise model that we name the \emph{\EFNM} (UPEF). 
This model (formally defined in Section~\ref{sec:upef-model:prelim}) 
still allows an unbounded number of corruptions determined by the adversary in an oblivious way. However, when the $i$-th transmission is corrupted by the adversary, the \emph{effect} of the corruption is random: the transmitted bit is flipped with some probability $p_i$ or erased with the complementary probability. The probabilities $\{p_i\}_{i\in \mathbb{N}}$ are parameters of the model and can be determined by the algorithm's designer.
In a sense, this type of noise matches the effect oblivious noise has on messages that are protected with a signature: with some probability the corruption is detected (the signature does not verify) and the message is marked as corrupted, i.e., erased. On the other hand, with some small probability the corruption is such that the signature verifies the corrupted message; in this case we have a flip.
The probabilities are determined by the length of the signatures in use.\footnote{We use AMD codes~\cite{Cra08} to generate signatures, 
see Appendix~\ref{apdx:uf}  for the exact details.}

This novel randomized model is much simpler to handle, and facilitates the design and analysis of optimal coding schemes. 
Furthermore, any scheme designed for this model can be translated into a coding scheme that works in the standard unbounded flips (UF) noise model, by employing signatures of respective length to match the erasure--flip probabilities of the UPEF model. Therefore, this model serves as a crucial tool for obtaining optimal coding schemes in the standard model.

Switching to the UPEF model allows us to use the scheme in~\cite{GI19} as the inner code of our concatenated coding scheme, in an almost as-is fashion: 
by smartly setting the probabilities~$\{p_i\}_{i\in \mathbb{N}}$, 
we can guarantee, with very high probability, that any execution experiences an unbounded number of erasures but only a bounded number of bit-flips.
The scheme in~\cite{GI19} withstands the erasures and delivers the non-erased bits (either correct or not)  to the outer layer, which should be able to cope with this limited amount of bit flips.

One problem still remains, but to explain it, we must first explain how the \cite{GI19}~scheme works. 
In a nutshell, the parties simulate the underlying (noiseless) protocol~$\pi$ bit-by-bit, where the scheme adds to each bit the parity (mod 2) of the corresponding round number. The scheme works in \emph{challenge-response}--style iterations: The first party (Alice) begins by sending the bit of the next round of~$\pi$ along with the parity of that round (\emph{the challenge}). 
Bob receives this bit, and if the parity corresponds to the round number he expects, he records this bit and replies with the next bit of~$\pi$ along with the parity of that round in~$\pi$ (\emph{the response}). When this reply reaches Alice, and the parity is correct, Alice records the bit from Bob and moves on to the next iteration. 
In any case of erasures or if the parity mismatches, the receiver ignores the received message. 
The analysis in~\cite{GI19} shows that this single parity bit suffices to keep track of the progress despite an unbounded amount of erasures.

However, in the UPEF model, a bit flip can either corrupt the content bit (i.e., the next simulated round of~$\pi$) or the parity bit sent along! 
Corrupting the parity bit damages the correctness of the \cite{GI19}~scheme, but this is the only way the noise can affect correctness.
Throughout a detailed case analysis, we prove that corrupting the parity bit has the sole effect of making the parties 
\emph{out-of-sync}, in the sense that one party advances to the next round in~$\pi$, while the other does not.
Luckily, this type of out-of-sync 
corruption was already considered in the interactive-coding community, initiated by the work of Braverman, Gelles, Mao, and Ostrovsky~\cite{BGMO17}, which presented a non-efficient scheme that withstands a noise level of up to $1/18$ fraction of the rounds, where ``noise'' here means insertions and deletions producing out-of-sync events as described above. 
That work was followed by a work by Sherstov and Wu~\cite{SW19}, who showed that a variant of the \cite{BGMO17}~scheme withstands the optimal level of noise, namely, up to $1/4$ of the rounds,
and by a work by Haeupler, Shahrasbi, and Vitercik~\cite{HSV18}, who presented an \emph{efficient} scheme, based on synchronization strings, with noise resilience of~$1/44$.
\looseness=-1

Therefore, we can set the scheme in~\cite{HSV18} (denoted HSV hereinafter) as the outer layer in our construction, and set the probabilities~$\{p_i\}_{i\in \mathbb{N}}$ such that the total number of insertion and deletion errors
will not surpass the threshold expected by the HSV scheme, except with an exponentially small probability.
This construction achieves our goal of obtaining a coding scheme in the UPEF model, with optimal length of $O(N+T)$, and an exponentially small failure probability. 

\smallskip

Unfortunately, once converting this optimal UPEF scheme back to the standard UF model, the overhead increases severely. 
In particular, the way we set the probabilities $\{p_i\}_{i\in \mathbb{N}}$ implies a logarithmic overhead on the size of the signatures, leading to a sub-optimal scheme of length $O((N+T)\log(N+T))$ in the UF model.
To avoid this increase in communication, we must maintain the probabilities $\{p_i\}_{i\in \mathbb{N}}$ ``large'', and design a new scheme that is still optimal in the UPEF model despite the high values of~$\{p_i\}_{i\in \mathbb{N}}$.

\paragraph{Obtaining an optimal scheme: the iterative approach.}
In order to obtain a UF-optimal scheme, we take a different approach, 
namely, we execute an increasing-length instances of a ``standard'' interactive coding~\cite{DHMSY18}.
As before, we start by constructing a coding scheme over the UPEF model. Our goal now is to maximize the $p_i$'s as much as possible.
The main idea is as follows. Let's fix each $p_i$ to some constant, say,~$2/3$. Now, any corruption in the UPEF model will cause an erasure with a fixed probability of~$1/3$. The number of erasures a party observes is a good estimate of the level of noise during the same transmissions. Hence, the parties can continue running the scheme again and again, until they believe the noise level was low enough to produce the correct output.

In more detail,
Alice and Bob run an efficient interactive coding scheme  resilient to a constant fraction of adversarial flips (e.g., \cite{BKN14,GH14}). After executing the scheme, Alice and Bob count the number of erasures observed during the execution and estimate (with high probability) the fraction of corruption they experienced. They communicate this estimate to each other, and decide how to continue accordingly.
If the noise level seems sufficiently low, the resilient scheme must have produced the correct output, and the parties can safely terminate. 
Otherwise, the parties re-run the interactive scheme, doubling its length. They repeat this action until they reach an execution where the noise level is low enough to guarantee the success of the underlying interactive coding scheme. 

With a correct choice of parameters, this results in a UPEF scheme of length $O(N+T)$. However, since all $\{p_i\}_{i\in \mathbb{N}}$ are fixed to a constant, once we translate this scheme into a UF scheme, we keep its length up to a constant and obtain an optimal length of $O(N+T)$ in this case as well.

We note that a communication complexity of~$\Theta(N+T)$ is tight for the UF model. 
A lower bound of~$\Omega(N+T)$ is immediate by considering the case where the adversary corrupts the entirety of the communication between Alice and Bob for $\Theta(T)$ rounds, e.g., by flipping each bit with probability 1/2, thereby not allowing any information to cross the channel during these rounds. After this corruption, $N$ rounds are still needed to complete the protocol without noise.

\subsection{Related Work}
As mentioned above, the field of interactive coding was initiated by the work of Schulman~\cite{schulman92,schulman96}. Following this work, many two-party interactive coding schemes were developed, with the goal to optimize various properties, such as  efficiency, communication rate, and noise resilience~\cite{BR14,GMS14,BKN14,KR13,haeupler14,GH15,BE17,GHKRW18,GZ22}. 
Two-party coding schemes for different types of noise, such as erasures or insertions and deletions, appeared in~\cite{FGOS15,GH15,EGH16,BGMO17,HSV18,SW19,EHK20,GZ22}. See~\cite{gelles17} for an extensive survey on this field.

Closest to our work are coding schemes that withstand an unbounded amount of corruption. 
As mentioned above,
Dani et al.~\cite{DHMSY18} developed a randomized scheme that deals with an unbounded amount~$T$ of oblivious bit-flips,  succeeds with high probability, and simulates any $\pi$ of length~$N$ in~$N+ O(T+\sqrt{N(T+1)\log T})$ rounds. 
Gelles and Iyer~\cite{GI19} developed a deterministic scheme that deals with an unbounded amount~$T$ of (not necessarily oblivious) erasures in at most $2N+4T$ communication rounds. 
For the multiparty setting, Aggarwal, Dani, Hayes, and Saia~\cite{ADHS20} developed a coding scheme that correctly simulates any protocol over an arbitrary network withstanding an unbounded amount of oblivious corruptions in $\tilde O(N+T)$ rounds, suppressing logarithmic terms.

\subsection{Paper Outline}
\label{sec:outline}
We set up the UPEF and UF models, recall the insertion-deletion model, and review interactive coding protocols in Section~\ref{sec:prelim}.
In Section~\ref{sec:UPEF-schemes-concat} we describe an optimal  UPEF coding scheme that follows a code concatenation approach. Its analysis is deferred to Appendix~\ref{apdx:analysis}. Finally, in Section~\ref{sec:iterative} we describe an optimal coding scheme in the UPEF model that follows an iterative approach. We then show how to translate it into an optimal UF coding scheme.

% =========================================================
\section{Preliminaries}\label{sec:prelim}
\paragraph{Notations.}
For an integer $n\in \mathbb{N}$ we use $[n]=\{1,2,3,\ldots, n\}$. All logarithms are taken to base~2 unless otherwise mentioned. For two strings $a,b$ we denote by $a\circ b$ their concatenation. We will use $\bigcirc_{k=1,2,...,\ell}\ a_k \triangleq a_1 \circ a_2 \circ \cdots \circ a_\ell$ to  abbreviate the concatenation of multiple strings. 
We use $O_\eps(\cdot), \Theta_\eps(\cdot)$, etc., to explicitly remind that the constant inside the $O(\cdot)$ may depend on (the constant)~$\eps$.

\subsection{Interactive Protocols and Coding Schemes}
Consider two parties, Alice and Bob, having inputs $x, y \in \{0, 1\}^k$ respectively, who wish
to compute some function~$f(x, y)$ by communicating over a channel with alphabet~$\Sigma$. 
Towards that goal, Alice and Bob use an \emph{interactive protocol} composed of two algorithms $\pi = (\pi_a, \pi_b)$ for Alice and Bob, respectively.
These algorithms assume a common clock known by both parties (i.e., the protocol is synchronized) and  determine, for each party in each round (timestep), whether the party 
(1) has to send a message in that round, 
(2) which symbol the party sends, and (3) if the party should terminate in that round and which output should it give.

Each party records all the messages it receives during the execution of the protocol. The collection of these records is the party's \emph{transcript}. 
We assume that $\pi$ has \emph{fixed order of speaking}; this means that in each round exactly one party is transmitting a symbol (the other party listens), and the identity of the transmitting party in a given round is predetermined and independent of the parties' inputs. 
In particular, a protocol in which Alice speaks in odd rounds, and Bob speaks in even rounds is said to be of an \emph{alternating order}. 
Note that if $\pi$ is not alternating, then it can be converted to an alternating-order protocol while increasing the communication complexity by a factor of at most 2. 
We say that a protocol is \emph{$k$-alternating}, for some $k\in\mathbb{N}$,
if during its execution each party transmits bulks of $k$ bits.
The \emph{length} of a protocol is defined to be the number of rounds it includes until both parties have terminated.

\paragraph{Noisy channels and coding schemes.}
Now, assume that the parties are connected by a \emph{noisy channel}. Formally,  given an input and output alphabets $\Sigma_{in}, \Sigma_{out}$, respectively, a single utilization of a noisy channel is the (possibly randomized) function $C: \Sigma_{in} \rightarrow \Sigma_{out}$. 

We can now discuss protocols that perform over noisy channels. 
We say that a protocol $\pi'$ \emph{simulates}~$\pi$ over the noisy channel~$C$, if for any inputs $(x,y)$, after executing $\pi'$ over the noisy channel~$C$, 
the parties can output their transcripts in an execution of~$\pi$ over a noiseless channel with inputs~$(x,y)$.
When the channel noise or the algorithm~$\pi'$ are probabilistic, we say that $\pi'$ simulates $\pi$ \emph{with probability~$p$} if the probability that the parties' output equal the transcript of~$\pi$ is at least~$p$, for any inputs pair. 

\medskip
A \emph{coding scheme} (for some given noisy channel~$C$) is a function $CS$, whose input is a noiseless protocol~$\pi$, and its output is a protocol $\pi' = CS(\pi)$ which simulates $\pi$ over the channel~$C$. 
When the channel noise or the scheme are probabilistic, we say that the coding scheme has success probability~$p$, if for any $\pi$, the protocol $\pi' = CS(\pi)$ simulates $\pi$ with probability~$p$.

\subsection{Noise Models}
As alluded to in the introduction, our scheme is designed to withstand an unbounded amount of (oblivious) bit flips. However, we design the scheme by reducing the unbounded-flip model to a different noise model with unbounded probabilistic erasures and flips. 
Furthermore, the effect of probabilistic erasures and flips noise on the inner layer of our coding scheme is such that the outer layer ``sees'' insertion and deletion noise. 
We will now define these three noise models in turn.

\paragraph{The Unbounded Flip  noise model (UF)}
\label{sec:uf-model:prelim}
Our main noise model is the unbounded flip noise model, set forth by Dani et al.~\cite{DHMSY18}. Given a specific execution of $\pi'$ with inputs~$(x,y)$, the adversary sets a \emph{noise corruption pattern} $E \subset \mathbb{N}$ such that the amount of noise, $|E|$, satisfies $|E|=T$ for some number $T\in\mathbb{N}$ decided by the adversary. The noise pattern can be set as a function of $\pi,x,y$ but is independent of any randomness the parties might have (i.e., an \emph{oblivious} noise). The noise pattern~$E$ determines which bits get flipped during the execution of~$\pi$. Namely, if $i \in E$, then the $i$-th transmitted bit in~$\pi$ will be flipped. Otherwise, the bit goes through uncorrupted. Note that $T$ might be arbitrary. When one of the parties terminates, the channel sends zeros to the another party, which may be flipped by the adversary.

\paragraph{The Unbounded Probabilistic Erasure-Flip noise model (UPEF)}
\label{sec:upef-model:prelim}
Our coding scheme in this work is designed and analyzed within the following noise model, that combines both erasure and flip noise. This model naturally appears when executing a protocol in the UF model while each message contains a (probabilistic) signature or a message authentication tag that indicates its validity.

In this model, the parties are connected via a noisy communication channel $C: \{0,1\} \to  \{0,1,\bot\}$, which can either flip bits or erase them (denoted by the erasure mark~$\bot$).
Similar to the UF model, 
given any specific execution of~$\pi'$ with inputs~$(x,y)$, the rounds which are corrupted are predetermined by an adversary that knows $\pi',x,y$ and the inputs but not the parties' private randomness.
This corruption is described via the noise pattern $E \subset \mathbb{N}$, where $i \in \mathbb{N}$ means that the $i$-th round is corrupted; otherwise, the bit arrives at the other side intact. 
When a round is corrupted, the effect is as follows:
the bit is flipped with some probability~$p_i$ or is erased
with probability~$1-p_i$. The probabilities $\{p_i\}_{i\in \mathbb{N}}$ are parameters of the model and will be specified later.

Terminating in the UPEF is different from terminating in the UF. When Alice terminates, the channel transmits a special ``silence'' symbol, namely, `$\square$'. Upon the reception of this special symbol, Bob knows that Alice has quit, and terminates as well.

Similar to the UF model, we restrict the discussion to noise patterns in which the total number of corrupted rounds is finite. That is, there exists some number $T\in\mathbb{N}$, unknown to the parties and~$\pi$, such that~$|E|=T$.

\paragraph{The Insertion-Deletion noise model}
\label{sec:id-model}
The insertion-deletion noise model~\cite{BGMO17}, which we briefly describe here, is important for our analysis of the concatenated coding scheme.

In this model we consider  \emph{alternating} interactive protocols~$\pi'$, where no common clock is assumed by the parties. 
Instead, Alice sends the first symbol (round 1), and Bob is idle until receiving this symbol. Once the first symbol is obtained by Bob he transmits a symbol back to Alice (round 2). Alice will execute round~3 once receiving this symbol, and so on.
The noise is allowed to either corrupt a symbol (i.e., the receiver will obtain a different symbol from the one sent, \emph{a substitution}), or to completely delete the symbol, so that the receiver receives nothing. In the latter case, the protocol is ``stuck'' as both parties await an incoming symbol to proceed. 
To avoid getting stuck, the noise must inject a new symbol towards the sender of the symbol that got deleted. This causes the parties to get \emph{out of sync}, that is, one of them will believe the current round is $i$, while the other will believe the current round is $i+2$.
See also~\cite{BGMO17,SW19,HSV18,EHK20,GKR21,GKR22}.

In \cite{HSV18}, the authors give an efficient constant-rate coding scheme for insertions and deletion noise, which we will use in our construction. 
\begin{theorem}[Theorem 1.2 in \cite{HSV18}]
\label{thm:HSV18}
For any alternating protocol~$\pi$ of length $n$ and for any $\eps>0$, there exists an efficient randomized protocol~$\pi'$ simulating $\pi$ in presence of $\delta = 1/44-\eps$ fraction of edit-corruptions, whose length is $\Theta_\eps(n)$ and succeeds with probability $1 - 2^{-\Theta(n)}$. The alphabet size of $\pi'$ is~$\Theta_\eps(1)$.
\end{theorem}
We assume that at its termination, $\pi'$ has an output, which equals to the output  of $\pi$ (under the conditions in the theorem).

% =========================================================

\section{A UPEF-optimal coding scheme via code concatenation}
\label{sec:UPEF-schemes-concat}
In this section we give an optimal UPEF coding scheme, based on code-concatenation approach (Algorithms~\ref{alg:schemeALICE} and~\ref{alg:schemeBOB}).
The analysis of the scheme, presented in the appendix, proves the following Theorem.
\begin{theorem}\label{thm:main}
Given any two-party binary interactive protocol $\pi$ of length $N$, there exists some constant $C$ and an efficient randomized protocol $\Pi$ of length $O(N+T)$ that simulates $\pi$ 
with probability~$1-2^{-\Omega(N)}$
over a binary channel in the presence of an arbitrary and a priori unknown number $T$ of probabilistic Erasure--Flip corruptions, with $p_i=\min\left\{\frac{CN}{i^2}, \frac{1}{2}\right\}$.
\end{theorem}

Let $\pi$ be a binary alternating protocol that assumes noiseless channels of length $|\pi|=N$. Our goal is to simulate $\pi$ in the UPEF model.
Let~$\pi'$ be the randomized protocol obtained from~$\pi$ via Theorem~\ref{thm:HSV18}, by setting $\delta=1/45$ (i.e., $\eps=1/1980$). We denote $|\pi'| = N'$.
Denote the alphabet of $\pi'$ by $\Sigma'$ and note that its size is constant, $|\Sigma'|=\Theta(1)$.

Our
coding scheme (Algorithms~\ref{alg:schemeALICE} and~\ref{alg:schemeBOB}) simulates the communication of~$\pi'$, symbol by symbol. As the channel in the \EFNM is binary, the parties communicate the binary representations of the symbols in $\pi'$. Therefore, during each iteration of the simulation, Alice sends a symbol to Bob using $\log|\Sigma'|$ bit transmissions, and expects a symbol reply from Bob.

\begin{algorithm}[ht]
\caption{Simulation over Erasure and Substitution Channel with Unbounded Noise (Alice)}
\label{alg:schemeALICE}
\algrenewcommand{\alglinenumber}[1]{\scriptsize A.{#1}}
    \makeatletter
	\renewcommand{\theALG@line}{A.\arabic{ALG@line}}
	\makeatother
	\begin{algorithmic}[1]
            \Statex \small \textbf{Input:} {An alternating binary protocol $\pi$ of length $N$,  an input $x$}
            \Statex \small \textbf{Initialize:} {Let $\pi'$ be the protocol simulating $\pi$ given by Theorem~\ref{thm:HSV18}, setting $\eps = 1/1980$. Let $N'=|\pi'|$ and assume $\Sigma'$ (the alphabet of $\pi'$) is a power of two.}
        \Statex
	    \State{\(\T_{a} \leftarrow \emptyset, r_a \leftarrow 0\)}
	    \While{\(r_a < \frac{N'}{2}\)} \label{line:alice:while}
	        \State{\texttt{// Send Message}}
            \State{\(r_a \leftarrow r_a + 1\)} \label{line:alice:incRA}
            \State{\(m_{send} \leftarrow \pi'(x \mid \T_a)\)} \label{line:alice:compute}
            \State{\(\T_a \leftarrow \T_a \circ m_{send}\)} \label{line:alice:insertAlice}
            \State{send \((m_{send}, r_a \mod 2)\)}   \Comment{$k$ transmissions }\label{line:alice:send}
            \State{}
            \State{\texttt{// Receive Message}}
            \State{receive \(m' = (m_{rec}, r_{rec})\)}  \Comment{$k$ transmissions }\label{line:alice:receive}
            \If{\(m'\) does not contain \(\perp\) and \(r_{rec} = r_a \mod 2\)} \label{line:alice:if}
                \State{\(\T_a \leftarrow \T_a\circ m_{rec}\)} \label{line:alice:insertBob}
            \Else
                \State{delete the last symbol of \(\T_a\)} \label{line:alice:delete}
                \State{\(r_a \leftarrow r_a - 1\)} \label{line:alice:decRA}
            \EndIf
	    \EndWhile   \label{line:alice:end-while}
            \State{Output the output given by $\pi'$}   \label{line:alice:output}
	\end{algorithmic}
\end{algorithm}

% =-=-=-=-=-=-=-=-=-=-=-=-=-=-=-=-=-=-=-=-=-=-=-=-=-=-=-=-=-=-=-=-=-=-=-=-
%%% Bob's side

\begin{algorithm}[ht]
\caption{Simulation over Erasure and Substitution Channel with Unbounded Noise (Bob)}
\label{alg:schemeBOB}
	\begin{algorithmic}[1]
	\algrenewcommand{\alglinenumber}[1]{\scriptsize B.{#1}}
	\makeatletter
	\renewcommand{\theALG@line}{B.\arabic{ALG@line}}
	\makeatother
            \Statex \small {\textbf{Input:} An alternating binary protocol $\pi$ of length $N$,  an input $x$}
            \Statex \small \textbf{Initialize:} {Let $\pi'$ be the protocol simulating $\pi$ given by Theorem~\ref{thm:HSV18}, setting $\eps = 1/1980$. Let $N'=|\pi'|$ and assume $\Sigma'$ (the alphabet of $\pi'$) is a power of two.}
        \Statex
	    \State{\(\T_{b} \leftarrow \emptyset, r_b \leftarrow 0, 
	    err \leftarrow 0, m \leftarrow (0, 0)\)}
	    \While{\(m' \neq \square\)}\label{alg:bob:while}
    	    \State{\texttt{// Receive Message}}
            \State{receive \(m' = (m_{rec}, r_{rec})\)}\Comment{$k$ transmissions } \label{line:bob:receive}
            \If{\(m'\) does not contain \(\perp\) and  \(r_{rec} \neq r_b \mod 2\)} \label{line:bob:if}
                \State{\(\T_b \leftarrow \T_b\circ m_{rec}\)} \label{line:bob:insertAlice}
                \State{\(err \leftarrow 0\)}
            \Else
                \State{\(err \leftarrow 1\)} \label{line:bob:err}
            \EndIf
            \State{}
            \State{\texttt{// Send Message}}
            \If{\(err = 0\)}
                \State{\(r_b \leftarrow r_b + 1\)}  \label{line:bob:increaseRB}
                \State{\(m_{send} \leftarrow \pi'(y \mid \T_b)\)} \label{line:bob:compute}
                \State{\(\T_b \leftarrow \T_b \circ m_{send}\)} \label{line:bob:insertBob}
                \State{send \(m \leftarrow (m_{send}, r_b \mod 2)\)}\Comment{$k$ transmissions } \label{line:bob:send}
            \Else
                \State{send \(m\)}\Comment{$k$ transmissions; $m$ from memory } \label{line:bob:previous}
            \EndIf
        \EndWhile
        \State{Output the output given by $\pi'$}
	\end{algorithmic}
\end{algorithm}
% =-=-=-=-=-=-=-=-=-=-=-=-=-=-=-=-=-=-=-=-=-=-=-=-=-=-=-=-=-=-=-=-=-=-=-=-

The simulation of~$\pi'$ employs a challenge-response paradigm, where Alice sends a symbol (the challenge) and expects one back (the response). 
The parties maintain  a counter to track their respective progress, namely, the variables $r_a$ and $r_b$, which represent the number of successful iterations observed by Alice and Bob, respectively.
Every time Alice and Bob send a symbol, they attach to it the parity (mod~2) of their own counter. Hence, the simulation is $k$-alternating, with $k=\lceil\log{|\Sigma'|}\rceil + 1$.

When Alice receives a symbol (as a response to the challenge she has previously sent), she checks the counter value attached to it: if it matches her expected counter parity (mod 2),  she ``believes'' this challenge-response iteration, delivers the received symbol to~$\pi'$ and increases $r_a$ by 1; otherwise, she ignores the reply and tries again in the next iteration. 

Bob acts in an analogous manner: if the information received from Alice matches the counter parity (mod 2) he is expecting, then he ``believes'' the received symbol, 
delivers it to~$\pi'$, increases $r_b$ by~1, obtains from $\pi'$ the next symbol to communicate to Alice, and sends Alice this symbol and the parity of~$r_b$. 
If the information from Alice does not match Bob's expectation, he ignores this transmission and replies with the previous symbol computed by~$\pi'$ (along with the parity that corresponds to that symbol).  
When a party ``believes'' an iteration, it appends the received and transmitted symbols of the iteration to its transcript, $\T_a$ or~$\T_b$, respectively. This transcript records all the symbols communicated so far (by~$\pi'$) during the ``successful'' iterations of Algorithms~\ref{alg:schemeALICE} and~\ref{alg:schemeBOB}. 

To summarize, 
in each iteration of the loop, Alice generates the next message of~$\pi'$, denoted $m\in \Sigma'$, based on her current transcript~$\T_a$ and her input~$x$, i.e.,  $m=\pi'(x \mid \T_a)$.
Alice (temporarily) adds $m$ to~$\T_a$, and sends its binary representation to Bob, along with the parity of~$r_a$. 
After receiving a $k$-bit message $(m_{rec},r_{rec})$ from Alice, Bob checks that none of the $k$ bits have been erased (denoted by~$\bot$) and that $r_{rec}$ is opposite to his parity (since Alice added a new symbol and he did not, yet). 
If everything matches, Bob adds $m_{rec}$ to his transcript~$\T_b$, increases $r_b$ by 1, computes the next message $m'=\pi'(y\mid \T_b)$ and the new parity of $r_b$, and transmits them to Alice. 
On the other hand, if Bob notices any erasures or the mismatch of the parity, he ignores Alice's new symbol and replies with the latest computed $(m',r_b)$ recorded in his memory.

At the end of the iteration, Alice receives a message and a parity from Bob; if there were no erasures and the received parity matches~$r_a$, she adds that message to her transcript. Otherwise, Alice deletes the (temporary) message she added at the beginning of this iteration.

The algorithm ends once the length of~$\T_a$ reaches the length of~$\pi'$ (for Alice) or when a special symbol~$\square$, sent by the channel when Alice quits, is received by Bob. We discuss  termination in this model and the implication of assuming this special symbol  in~Appendix~\ref{sec:termination}.

\smallskip

The complete detailed analysis of the coding scheme (Algorithms~\ref{alg:schemeALICE} and~\ref{alg:schemeBOB}) and the proof of Theorem~\ref{thm:main}, are deferred to Appendix~\ref{apdx:analysis}.

% =========================================================

\section{A UF Scheme with Optimal Communication}
\label{sec:iterative}
The concatenated scheme in Section~\ref{sec:UPEF-schemes-concat}, while optimal in the UPEF model, 
implies a UF scheme  of length $O((N+T)\log(N+T))$
in which the $i$-th bit is replaced with an AMD code~\cite{Cra08} of length $O(\log (p_i^{-1}))$; see Appendix~\ref{apdx:uf} for details. 
In this section, we take a different approach towards constructing an optimal coding scheme in the UPEF model, namely by executing increasing-length coding schemes in an iterative fashion. This approach eventually leads to an optimal-communication UF scheme.

Here, 
Alice and Bob utilize a ``standard'' interactive coding scheme for substitutions and estimate the experienced level of noise. 
If the estimated noise level is too high, Alice and Bob repeat the execution with a larger amount of redundancy. 
When the noise level is low enough, the interactive scheme guarantees the success of the computation, and the parties can terminate.
This approach, in addition to leading to a UF scheme with optimal communication, is also much simpler and easier to analyze than the scheme of Section~\ref{sec:UPEF-schemes-concat}. 
The key difference is that, all we need in order to estimate the noise level well, is that the probabilities $\{p_i\}_{i\in \mathbb{N}}$ are bounded below by a constant, rather than converging to~$0$. This aligns perfectly with obtaining optimal complexity, as smaller $\{p_i\}_{i\in \mathbb{N}}$ imply longer encodings.

The scheme in this section utilizes a slight variant of the UPEF model, denoted the modified UPEF (mUPEF) model, which we now describe.
Let a noise pattern $E \subset \mathbb{N}$ be determined adversarially.
As before, if $i \notin E$, the $i$-th transmitted bit reaches the other party intact. 
For $i \in E$, the $i$-th transmitted bit is still erased with 
probability~$1-p_i$, and corrupted with probability~$p_i$.
However, the corruption here is not necessarily a bit flip as in the original UPEF. Instead, the adversary determines whether the bit is flipped, erased, or not corrupted at all.
The probabilities~$\{p_i\}_{i\in\mathbb{N}}$ are parameters of the model. However, we will actually set them all to have the same value. That is, $\forall i,\ p_i = 2/3$. 

In our scheme, we set $p_e = 1-p_i= 1/3$ as a lower bound on the probability that a bit is erased.
Similar to the UF model, we will assume that when Alice terminates, the channel implicitly sends Bob a default symbol (e.g., a zero). 

\smallskip
In the following sections we describe and analyze a coding scheme in the mUPEF model, with optimal $O(N+T)$ communication that, due to our choice of $p_e=2/3$, results with a UF scheme with $O(N+T)$ communication as well. 

\paragraph{A optimal mUPEF Coding Scheme.}
For the underlying substitution-resilient interactive coding scheme, we can take any (efficient) 2-party scheme with binary alphabet that is resilient to a constant fraction of adversarial noise, e.g.,~\cite{BKN14,GH14}.
In particular, let us assume an interactive coding scheme that 
simulates any (noiseless) protocol $\pi$ in the presence of up to 0.1 adversarial substitutions with a constant rate over the binary alphabet. 
Denote the substitution-resilient version by~$\pi'$, so  $|\pi'|=O(|\pi|)=O(N)$.
We assume that $\pi'$ is alternating. 
We additionally assume that the communication of $\pi'$ includes a constant fraction of ones. To be concrete, out of the $|\pi'|/2$ bits Alice sends, at least $|\pi'|/8$ are the bit~1, in any execution of~$\pi'$.
We must have this property, because in our scheme, a long sequence of zeros will indicate termination. For that reason, we want that $\pi'$ will not send a long sequence of zeros. This can be achieved, for instance by making the parties communicate a~$1$ every alternate round, or by means of randomization (see, e.g., \cite{GH15}).

We proceed to describe our mUPEF resilient scheme~$\Pi$ simulating the substitution-resilient $\pi'$ defined above. The execution of~$\Pi$ consists of iterations, where the $i$-th iteration, $i=0,1,2,\ldots,$ takes $2L_i$ rounds with $L_i = |\pi'|2^i$. 
The $i$-th iteration can be broken down into two parts, each of length~$L_i$. 
In the first part, the parties execute $\pi'$ from scratch, padded to length $L_i$; in this padded protocol, each bit of $\pi'$ is sent $2^i$ times, and decoding is performed by majority (defaulting to 0 on ties, considering erased copies as zeros). 
In the second part of iteration $i$, only Bob speaks. He sends Alice the \emph{success string}~$0^{L_i}$ if and only if he observed less than $0.001p_eL_i$ erasures in the first part; otherwise Bob sends the \emph{error string}~$1^{L_i}$.

Alice terminates at the end of iteration~$i$ 
if she observed less than $0.001p_eL_i$ erasures in each of the parts of iteration~$i$, 
and Bob's transmissions at the second part contains more 0's than 1's (i.e., it decodes to the success string rather than to the error string). 
Alice gives as an output the same output that $\pi'$ has generated in the iteration in which she terminated. 
Bob terminates at the end of iteration~$j$ if  he observed less than $0.001p_eL_j$ erasures in the first part of $j$
and at most $0.001p_eL_j$ of the received bits in the first part of $j$ are 1's. 
Bob gives as an output the output of the latest iteration~$k$, with $k<j$,  
in which (1) he observed less than $0.001p_eL_j$ erasures in the first part and (2) he received at least $L_k/40$ ones from Alice in the first part. 
We call \emph{a valid iteration} any iteration that satisfies these two conditions. 

\subsection{Analysis}
In this section we prove the following theorem.
\begin{theorem}\label{thm:main:optimal}
Given any two-party binary interactive protocol $\pi$ of length $N$, there exists an efficient randomized protocol $\Pi$ of length $O(N+T)$ that simulates $\pi$ 
with probability $1-2^{-\Omega(N)}$
over a binary channel in the presence of an arbitrary and a priori unknown number~$T$ of mUPEF corruptions, with $p_i=2/3$ for all $i$.
\end{theorem}
We start with proving the correctness of our coding scheme:
we begin by demonstrating that Alice's output is correct with high probability. Additionally, we show that Bob's output at Alice's termination is also correct.
Then, we prove that Bob terminates after Alice, that Alice terminates in a \emph{valid} iteration, and that there are no valid iterations afterwards.
This would imply that 
Bob gives the right output as well.

Recall that $\pi'$ is resilient to $0.1$-fraction of substitutions. In addition, recall the padding mechanism; in order to cause a bit substitution in $\pi'$ in some iteration $i$, at least a half of its $2^i$ transmitted copies must be flipped or erased. 
Thus, if there are less than $L_i/20$ corruptions during the first part of iteration $i$ (that is, indices that the adversary puts in $E$),
then $\pi'$ must give the correct output at the end of iteration $i$, for both Alice and Bob.

In the following lemma, we show that if there are \emph{more} than $L_i/20$ corruptions
in some iteration~$i$, then Alice continues to executing iteration~$i+1$ and does not terminate at the end of iteration~$i$, except with a negligible probability of~$2^{-\Omega_{p_e}(L_i)}$.

\begin{lemma}
\label{lem:erasures-flips}
    Assume that in the first part of iteration~$i$ there are $c_i \ge L_i/20$ corruptions. Then, the probability that Alice terminates at the end of iteration~$i$ is~$2^{-\Omega_{p_e}(L_i)}$.
\end{lemma}

\begin{proof}
We divide the proof into two separate cases: (1) each of Alice and Bob observes less than $0.001p_eL_i$ erasures in the first part, and (2) Bob observes more than $0.001p_eL_i$ erasures but Alice does not.
Of course, if more than $0.001p_eL_i$ erasures are observed by Alice during the first part, she does not terminate by definition.

First, we find the probability of case~(1) to occur.
Denote by $E_i$ the subset of the noise pattern $E$ containing only rounds in the first part of iteration~$i$. 
Then,  $|E_i|\ge c_i \ge L_i/20$.
Let $e'$ be the number of rounds during the first part of iteration~$i$ that were erased because the event of channel erasure, which happens with probability~$p_e$, occurred.
It holds that
$\mathbb{E}[e'] = |E_i|p_e \ge 0.05p_eL_i$, and by Chernoff's inequality (Theorem~4.4(2) in~\cite{MU17}),
\begin{equation*}
\Pr(e' < 0.002p_eL_i) \leq \Pr(e' < 0.04\mathbb{E}[e']) \leq e^{-\mathbb{E}[e']\frac{0.96^2}{2}} \leq e^{-0.05p_eL_i\frac{0.96^2}{2}} \in 2^{-\Omega_{p_e}(L_i)}\text{.}
\end{equation*}

Thus, if $c_i \ge L_i/20$, then $e' < 0.002p_eL_i$ with probability $2^{-\Omega_{p_e}(L_i)}$, which bounds the probability of case~(1) to occur.

In case (2), Bob sees a lot of erasures and sends the error string. In order for Alice to terminate, it is necessary for her to decode this message as the success string.
For this to happen, it is necessary that the adversary corrupts at least $L_i/2$ bits during the second part of the $i$-th iteration. 
Similar to the proof of case~(1), the adversary succeeds to corrupt so many bits while causing less than $0.001p_eL_i$~erasures during the second part with probability of at most~$2^{-\Omega_{p_e}(L_i)}$. 
We conclude that Alice terminates at the end of iteration~$i$ with probability of at most~$2^{-\Omega_{p_e}(L_i)}$.
\end{proof}

The above lemma indicates that, in the event of an excessive number of corruptions, Alice will not terminate.
The following observation complements this idea and states that if Alice does terminate, the computation of $\pi'$ is successful with high probability, hence, her output in~$\Pi$ is correct.
\begin{observation}
\label{obs:correct-termination}
    When Alice terminates, $\pi'$ gives the correct output, with probability $1-2^{-\Omega_{p_e}(N)}$.
\end{observation}
\begin{proof}
If during a given iteration there were less than $L_i/20$ corruptions, then the resilience of~$\pi'$ guarantees that it gives the right output, and Alice terminates. 
The event in which Alice terminates and provides incorrect output can only occur when there are more corruptions in a specific iteration, the probability of which was constrained in Lemma \ref{lem:erasures-flips}.
A union bound on all possible iterations (while recalling  that $L_i=|\pi'|2^i$) bounds the probability for Alice to give an incorrect output, by 
\begin{equation*}
\sum_{i=0}^{\infty}{2^{-\Omega_{p_e}(L_i)}} = 2^{-\Omega_{p_e}(|\pi'|)} = 2^{-\Omega_{p_e}(N)}\text{.}\qedhere
\end{equation*}
\end{proof}

Next, we prove that Bob terminates only after Alice has already terminated, with high probability. 

\begin{lemma}
\label{lem:termination-after-alice}
Consider an iteration $i$ in which Alice has not yet terminated. 
Then, the probability that Bob terminates at the end of iteration~$i$ is at most $2^{-\Omega_{p_e}(L_i)}$.
\end{lemma}
\begin{proof}
In order to terminate at the end of iteration~$i$, Bob must observe less than $0.001p_eL_i$~erasures in the first part of the iteration. In a manner analogous to the argument presented in Lemma \ref{lem:erasures-flips}, it can be shown that, with a probability approaching $1-2^{-\Omega_{p_e}(L_i)}$, there will be a total of less than $L_i/20$ corrupted messages received by Bob during the first part of iteration $i$.
Recall that as long as Alice has not terminated, at least $L_i/8$ out of the $L_i/2$ bits she sends in iteration~$i$ are ones. Consider these transmissions only.
Even if all the corruptions during the first part of $i$ occur during these transmissions, then Bob will observe at least $0.001p_eL_i$ ones in the first part of $i$, and will therefore not terminate by definition.
Consequently, Bob terminates in iteration $i$ with probability~$2^{-\Omega_{p_e}(L_i)}$.
\end{proof}

Bob's output is the output of $\pi'$ in the last \emph{valid} iteration prior to his termination iteration.
The following lemma 
shows that, with high probability, the last valid iteration is the same iteration in which Alice has terminated. 
By Observation~\ref{obs:correct-termination}, $\pi'$ gives the correct output in that iteration.

\begin{lemma}
\label{lem:valid-iterations}
    Denote by $i$ the iteration in which Alice terminates. Then, $i$ is a valid iteration with probability $1-2^{-\Omega_{p_e}(L_i)}$.
    Further, the probability that there is a valid iteration $j>i$ is $2^{-\Omega_{p_e}(N)}$.
\end{lemma}
\begin{proof}
In order to prove that $i$ is a valid iteration, we have to show that Bob observes less than $0.001p_eL_i$ erasures in the first part of $i$, and that he receives at least $L_i/40$ ones from Alice in the first part. 

First, note that Bob observes less than $0.001p_eL_i$ erasures in the first part if and only if he sends the success string in the second part. We demonstrate that with probability $2^{-\Omega_{p_e}(L_i)}$ this is the case. As illustrated in case (2) of Lemma \ref{lem:erasures-flips}, when Bob transmits the error string to Alice in the second part of $i$, Alice terminates with probability $2^{-\Omega_{p_e}(L_i)}$. Since Alice terminates in $i$, there is a probability of $2^{-\Omega_{p_e}(L_i)}$ that Bob sends the error string in the second part of $i$ and observes more than $0.001p_eL_i$ erasures in the first part of $i$.

We proceed to show that Bob receives at least $L_i/40$ ones from Alice in the first part of $i$. By Lemma \ref{lem:erasures-flips}, during the first part of iteration~$i$ there are less than $L_i/20$ corruptions with probability $1-2^{-\Omega_{p_e}(L_i)}$. 
Additionally, Alice always sends at least $L_i/8$ ones in the first part. Thus, Bob receives at least $L_i/40$ ones during the first part of $i$ with probability $1-2^{-\Omega_{p_e}(L_i)}$, and this is the probability of $i$ to be a valid iteration.

Let $j$ be an iteration such that $j>i$. Recall that after Alice terminates the channel sends Bob zeros, by default.
We show that $j$ is not a valid iteration with high probability, by dividing into two cases. 
If there are at least $L_j/40$ flips in the transmissions towards Bob during the first part of iteration~$j$, then with probability $2^{-\Omega_{p_e}(L_j)}$ Bob  observes at least $0.001p_eL_j$ erasures, similar to case (1) in Lemma \ref{lem:erasures-flips}, thus $j$ is not valid. 
If there are less than $L_j/40$ flips in these transmissions, then Bob receives less than $L_j/40$ ones and $j$ is not a valid iteration either. 
Thus, iteration $j$ is valid with probability of at most $2^{-\Omega_{p_e}(L_j)}$. 
Since $L_j = |\pi'|2^j$, applying a union bound on all the iterations
gives a probability of $\sum_{j=i}^{\infty}{2^{-\Omega_{p_e}(L_j)}} = 2^{-\Omega_{p_e}(N)}$ to the event that there is a valid iteration after~$i$.
\end{proof}

We may use a union bound to conclude the correctness part for Bob. 
The overall probability of Bob to terminate after Alice is $1-2^{-\Omega_{p_e}(N)}$. The probability of Bob to declare Alice's termination iteration as valid is $1-2^{-\Omega_{p_e}(N)}$. 
Bob gives  the correct output at this iteration with probability $1-2^{-\Omega_{p_e}(N)}$ by Lemma \ref{lem:erasures-flips}, and this iteration is the last valid iteration with probability $1-2^{-\Omega_{p_e}(N)}$. We may use the inclusion-exclusion principle to show that the probability of the intersection of all these four events is $1-2^{-\Omega_{p_e}(N)}$. Recall that for any $A, B$ it holds that $1\ge P(A\cup B) = P(A) + P(B) - P(A\cap B)$.
Then, the fact that $P(A), P(B) \in 1-2^{-\Omega_{p_e}(N)}$ implies that $P(A\cap B) \ge 2(1-2^{-\Omega_{p_e}(N)})-1 \in 1-2^{-\Omega_{p_e}(N)}$. 
Thus, the probability of the event in which Bob gives the correct output at his termination is $1-2^{-\Omega_{p_e}(N)}$, and we have completed the correctness part of Theorem \ref{thm:main:optimal}.

\smallskip
The communication complexity of $\Pi$ is $\sum_{i=1}^{i_B}{2L_i}$, with $i_B$ being the iteration in which Bob terminates. 
At any iteration $i$ before Alice terminates, the adversary has to select at least $0.001p_eL_i$ bits to corrupt in order to prevent Alice from terminating. In addition, at any iteration $i$ after Alice's and before Bob's terminations, the adversary has to select at least $0.001p_eL_i$ bits to corrupt in order to prevent Bob from terminating. 
Denote the iteration in which Alice terminates by $i_A < i_B$. Then, $\sum_{i=1, i \ne i_A}^{i_B-1}{0.001p_eL_i} < T$, so $\sum_{i=1, i \ne i_A}^{i_B-1}{2L_i} < 2000/p_e \times T$. Since $L_i$ satisfies $L_i=2L_{i-1}$ for all $i$, and since $L_{i_A-1}, L_{i_B-1}$ are included in $\sum_{i=1, i \ne i_A}^{i_B-1}{2L_i} < 2000/p_e \times T$, then $\sum_{i=1}^{i_B}{2L_i} < 6000/p_e \times T$ and the communication complexity is $O_{p_e}(N+T)$.

\subsection{Obtaining a UF-optimal coding scheme}
In this section we construct a UF-model scheme based on the mUPEF protocol~$\Pi$, while maintaining a communication complexity of $O(N+T)$. Recall, we set $p_e = 1/3$. 
This means that adversarial corruption in the $i$-th transmission of $\Pi$, becomes detectable (an erasure) with probability at least~$1/3$. We would like to simulate this property in the UF model.

Towards this goal, we independently encode each bit of~$\Pi$ using a random code of length~5.
In particular, we encode a 0 to one of  $\{00000, 10000, 01000\}$ with equal probability,
and encode a 1 to one of $\{00100,10010,01001\}$ with equal probability. 
A received 5-bit word is decoded to a 0 or a 1 only if it belongs to respective set, or otherwise it is considered as an erasure. 
It can easily be seen that any pattern of 1 to~5 bit-flips decodes to an erasure with probability at least~$1/3$ as desired.
We have thus inflated the scheme by only a constant factor. The bit complexity of the resulting UF scheme is then~$5\cdot|\Pi|=O(N+T)$.

% =========================================================
\section*{Acknowledgments}
This work is partially supported by a grant from the United States-Israel Binational Science Foundation (BSF), Jerusalem, Israel, Grant No.\@ 2020277.
E. Fargion and R. Gelles  would like to thank Paderborn University for hosting them while part of this research was done.
R. Gelles would like to thank  CISPA---Helmholtz Center for Information Security for hosting him.

% =========================================================
% bibliography
\bibliographystyle{alphabbrv-doi}
\bibliography{coding}

% =========================================================
\appendix
\section*{\huge{Appendix}}
\section{An Analysis of the Code Concatenation Scheme}
\label{apdx:analysis}
In this section, we analyze Algorithms~\ref{alg:schemeALICE} and~\ref{alg:schemeBOB} and prove Theorem \ref{thm:main}.

{
\renewcommand{\thetheorem}{\ref{thm:main}}
\begin{theorem}
Given any two-party binary interactive protocol $\pi$ of length $N$, there exists some constant $C$ and an efficient randomized protocol $\Pi$ of length $O(N+T)$ that simulates $\pi$ 
with probability~$1-2^{-\Omega(N)}$
over a binary channel in the presence of an arbitrary and a priori unknown number $T$ of probabilistic Erasure-Flip corruptions, with $p_i=\min\left\{\frac{CN}{i^2}, \frac{1}{2}\right\}$.
\end{theorem}
\addtocounter{theorem}{-1}
}

We argue that for any~$\pi$ satisfying the conditions in \Cref{thm:main}, the protocol $\Pi$ obtained by applying Algorithms~\ref{alg:schemeALICE} and~\ref{alg:schemeBOB} on $\pi$ satisfies the theorem. 
In the following, we analyze the behavior of $\Pi$ (in the UPEF model), and reduce it to an execution of $\pi'$ in the insertion-deletion model with a related noise pattern. We show that this noise pattern exhibits the sufficiently low fraction of edit-corruptions required by Theorem~\ref{thm:HSV18} to guarantee the success of the execution of $\pi'$. Consequently, this implies the success of~$\Pi$.

Throughout this section we assume a fixed $\pi,x,y$. Let $\pi'$ and $\Pi$ be given by Algorithms~\ref{alg:schemeALICE} and~\ref{alg:schemeBOB}. When discussing executions of $\Pi$ and $\pi'$ we assume some arbitrary noise pattern that occurs during the execution, unless a pattern is explicitly specified.
We begin the analysis by setting forth several definitions and observations that will be useful throughout the analysis.

\begin{definition}[Iterations]
\label{def:iteration}
For any $i\in\mathbb{N}$, the consecutive rounds 
$[2(i-1)(\log|\Sigma'| + 1) + 1,$ $2i(\log|\Sigma'| + 1)]$
are called \emph{iteration}~$i$.
\end{definition}
The $i$-th iteration as defined above corresponds to the rounds in which Alice and Bob transmit during their $i$-th execution (iteration) of the while loop (Line~\ref{line:alice:while} and~\ref{alg:bob:while}, respectively).

The following observation is rather straightforward.
\begin{lemma}
\label{lem:progress}
In each iteration, each party adds either 0 or 2 symbols to its transcript.
\end{lemma}
\begin{proof}
Let us first count the symbols added by Alice during an iteration. Alice always adds one symbol to~$\T_a$ in Line~\ref{line:alice:insertAlice}. 
Then, according to the \emph{if} statement in Line~\ref{line:alice:if}, she either adds another symbol in Line~\ref{line:alice:insertBob}, or deletes one symbol in Line~\ref{line:alice:delete}. 
Next, consider Bob (Algorithm~\ref{alg:schemeBOB}); there are two cases here. In the first case,
Bob executes Line~\ref{line:bob:insertAlice} and adds Alice's symbol to~$\T_b$. In this case, $err=0$ and Bob adds another symbol in Line~\ref{line:bob:insertBob}. 
In the other case, Line~\ref{line:bob:insertAlice} is not executed and $err=1$, thus no further symbols will be added to~$\T_b$.
\end{proof}

\begin{definition}[Progress]
\label{def:progress}
When a party adds two symbols to its transcript in a certain iteration, we say that the party \emph{makes progress} in that iteration. Otherwise, we say that the party \emph{does not make progress}.
Further, we distinguish the following types of iterations:
\begin{enumerate}
    \item When both parties add the same two symbols to their transcripts, we say that the parties make the \emph{same progress}.
    \item When both parties make progress, but not the \emph{same} progress, we say that the iteration ends up with a \emph{substitution}.
    \item When only one party makes progress, we say that the iteration ends up with an \emph{insertion}. Further, the party that made the progress is the \emph{inserting party}.
\end{enumerate}
\end{definition}

An iteration in which at least one party makes progress will be designated as an iteration with progress. When neither party makes any progress in a given iteration, that iteration is said to have \emph{no-progress}. A party's \emph{state} is defined as $(\T_a, r_a)$ or $(\T_b, r_b, m)$, depending whether it is Alice or Bob. Note that the lines a party executes in some iteration and the messages it sends depend only on its state, its inputs and the message it receives from the other party during the iteration.

\subsection{Sequences}
\label{apdx:sequence:all}
The notion of a sequence, which we now define, 
is the basic unit of the simulation, which we can analyze in an organic manner.
A sequence contains consecutive iterations that begin at an iteration where Alice and Bob are ``synchronized'', and lasts until the next iteration they are synchronized again (which could be the immediately subsequent iteration, or several iterations later).

\begin{definition}[A sequence]
Assume $r_a=r_b$ at the beginning of some iteration $i$, and let $j \ge i$ be the first iteration with $r_a=r_b$ at its end. The consecutive series of iterations $[i,j]$ is called \emph{a sequence}.
\end{definition}

In this section we state and prove a lemma that describes the different types of sequences, and relates them to the minimal number of corruptions that can cause them.

First, note that the parties send messages in Lines~\ref{line:alice:send} and~\ref{line:bob:send} that has two parts,~$(m,r)$, which we refer to as ``the $m$ part'' and ``the $r$ part'', respectively, for short.
\begin{lemma}
\label{lem:sequence:all}
Let $[i,j]$ be a sequence. Then, one of the following holds:
\begin{enumerate}
    \item The sequence contains a single iteration, $i=j$, in which either (a) both parties make progress, and further, if there are no flips during this iteration, then both parties make the same progress; or, (b) both parties don't make progress, and there is at least one corruption.
    \item The sequence has more than one iteration, and the following holds:
    \begin{enumerate}
        \item At iteration~$i$ there is an insertion, and there is at least one corruption. If there are no flips in the $r$-bits in that iteration, then the inserting party is Bob.
        \item At each iteration~$k$, where $i<k<j$, either (a) both parties don't make any progress, or (b) both parties make progress and the noise flips both $r_a$ and~$r_b$. 
        In either case, there is at least one corruption in iteration $k$. Further, if there are no flips in the $m$ parts of the iterations and case (b) happens, the parties make the same progress.
        \item At iteration~$j$, only one of the parties makes progress. Further, if there are no flips in the $r$-bits in that iteration, then the party that makes progress is Alice. Additionally,
            if there are no flips in the sequence at all, then the progress Alice makes is equal to the progress Bob has made in iteration~$i$. 
    \end{enumerate}
\end{enumerate}
\end{lemma}

We prove lemma \ref{lem:sequence:all} in two steps. First, we restrict the discussion to the case where corruptions could affect only the the $r$~part, i.e., the bits~$r_a$ and~$r_b$ communicated by the parties. These bits are responsible for the synchronization between the parties and determine whether parties make progress or not. We then regard the case that the transmissions that correspond to the $m$~part might be corrupted as well. These are the ``content'' bits, so corrupting them translates to corrupting the ``data'' of the message.  Corrupting the data part cannot cause synchronization issues since whether a message is recorded or discarded by a party is decided solely by the $r_a,r_b$ bits.

Let us start with analyzing corruptions that occur only in the $r_a,r_b$ parts.
\begin{lemma}
\label{lem:sequence:manage}
    Assume the noise (flips and erasures) corrupts only transmissions when $r_a$ and $r_b$ are communicated. Then, Lemma~\ref{lem:sequence:all} holds.
\end{lemma}

\begin{proof}
In each iteration, the adversary has eight possible ways to corrupt the transmissions that contain the $r_a,r_b$ information (in addition to not corrupting at all): erasing only one or both of $r_a, r_b$, flipping only one or both of them, or erasing one of them and flipping the other one. 
We analyze the nine cases in turn.
We repeat the proof twice, once for the case  
where $r_a = r_b $ (mod~2) at the beginning of the iteration, and second time for iterations with $r_a \ne r_b$ (mod~2) at the beginning of the iteration.

\goodbreak
\noindent\begin{description}[left=0pt]
\item \textbf{The case where $r_a=r_b \pmod{2}$ at the beginning of the iteration:}
\begin{enumerate}[topsep=0pt,parsep=0pt]
    \item \textbf{No corruptions:}
    Alice raises $r_a$ (Line~\ref{line:alice:incRA}) and sends this value to Bob (Line~\ref{line:alice:send}); note that now $r_a\ne r_b \pmod 2$. Bob gets $r_a' \neq r_b \pmod 2$ and since there are no corruptions and $r'_a$ fits the expected value, Bob sets $err=0$ and thus he raises $r_b$ (Line~\ref{line:bob:increaseRB}), which is sent back to Alice (Line~\ref{line:bob:send}). Alice gets the expected value (Line~\ref{line:alice:if}) and thus does not change her $r_a$ anymore. We get that $r_a=r_b \pmod 2$ at the end of this iteration, and both parties make the same progress.

    \item \textbf{Only $r_a$ is flipped:} As in the previous case, Alice raises $r_a$ and sends this value to Bob; $r_a \neq r_b$, but Bob gets $r_a' = r_b \pmod 2$ so $r'_a$ does not fit the expected value. Thus, Bob sets $err=1$, does not raise~$r_b$ and sends Alice the message saved in his memory, which has value $r$ such that $r\ne r_a\pmod 2$ (Line~\ref{line:bob:previous}). Alice does not get the expected value (Line~\ref{line:alice:if}) and decreases her $r_a$ (Line~\ref{line:alice:decRA}). We get that $r_a=r_b$ (mod 2) at the end of this iteration, and none of the parties makes progress.
    \item \textbf{Only $r_a$ is erased:} This case is similar to the case of flipping only~$r_a$, but instead of getting the wrong~$r_a$, Bob gets $\bot$. Nevertheless, Bob and Alice perform exactly as in the previous case.
    \item \textbf{Both $r_a, r_b$ are erased:} This is similar to the case of erasing only~$r_a$. The only difference is that instead of getting the wrong~$r_b$, Alice gets $\bot$, but she executes the same code nevertheless.
    \item \textbf{$r_a$ is flipped,  $r_b$ is erased:} This is similar to the case of erasing both $r_a, r_b$. Here, instead of getting~$\bot$, Bob gets the wrong $r_a$, but he executes the same lines (and so is Alice).
\end{enumerate}
\qquad\tikz{
\draw (1,0) -- (\columnwidth,0);
\draw[fill=black] (1,0) circle (0.075);
\draw[fill=black] (\columnwidth,0) circle (0.075);
}
\begin{enumerate}[resume*,topsep=0pt,parsep=0pt]

    \item \textbf{Only $r_b$ is flipped:} Alice raises $r_a$ (Line~\ref{line:alice:incRA}) and sends this value to Bob (Line~\ref{line:alice:send}). Bob gets $r_a' \neq r_b \pmod 2$ and since there are no corruptions and $r'_a$ fits the expected value, Bob sets $err=0$ and thus he raises $r_b$ (Line~\ref{line:bob:increaseRB}), which is sent back to Alice (Line~\ref{line:bob:send}). Alice gets flipped $r_b$, which does not match to her expected value, (Line~\ref{line:alice:if}) and thus decreases $r_a$ in Line~\ref{line:alice:decRA}. We get that $r_a \neq r_b\pmod 2$ at the end of this iteration, and only Bob makes progress.
    \item\textbf{Only $r_b$ is erased:} Similar to the case of flipping only $r_b$. Here, instead of getting the wrong $r_b$, Alice gets $\bot$ and behaves as in the previous case.
    \item\textbf{$r_a$ is erased,  $r_b$ is flipped:} Alice raises $r_a$ (Line~\ref{line:alice:incRA}) and sends this value to Bob (Line~\ref{line:alice:send}). Bob gets $\bot$, sets $err=1$, and thus he sends Alice the message he sent in the previous iteration (Line~\ref{line:bob:previous}), and $r_a \neq r_b \mod 2$. Alice gets the flipped $r_b$, which fits to her expected value (Line~\ref{line:alice:if}) and thus adds the message from Bob to her transcript. We get that $r_a \neq r_b$ (mod 2) at the end of this iteration, and only Alice makes progress.
    \item\textbf{Both $r_a, r_b$ are flipped:} Similar to the case of erasing $r_a$ and flipping $r_b$. Instead of getting a~$\bot$, Bob gets the wrong $r_a$, but he executes the same lines (and also Alice).
\end{enumerate}
Note that cases (1)--(5) end up with $r_a=r_b\pmod{2}$, while cases (6)--(9) end up with $r_a\ne r_b\pmod2$.
In case (1) both parties make the same (correct) progress, while in (2)--(5) none of the parties makes progress. Thus, items (1)--(5) correspond to item~$(1)$ of the lemma's statement; indeed, when both parties don't make progress there is at least one corruption.

Items (6)--(9) correspond to item $(2)(a)$, where only one party makes a progress and $r_a \ne r_b \pmod2$ at the end of this iteration, so the sequence is not over.
Also note that in the case where there are no flips (but there are erasures, i.e., item~(7)), Bob is the party that makes progress. In all cases (6)-(9) there is at least one corruption.

\goodbreak
\item\textbf{The case where $r_a\ne r_b \pmod{2}$ at the beginning of the iteration:}
\begin{enumerate}[topsep=0pt,parsep=0pt]
    \item\textbf{Only $r_b$ is erased:} Alice raises $r_a$ (Line~\ref{line:alice:incRA}), and $r_a=r_b\pmod2$. Thus, Bob gets $r_a$ which does not fit to the expected, and sends Alice the message he sent in the previous iteration (Line~\ref{line:bob:previous}). Alice gets~$\bot$, decreases~$r_a$ (Line~\ref{line:alice:decRA}) and at the end of the iteration $r_a \neq r_b \pmod 2$ and none of the parties make progress.
    \item\textbf{Both $r_a,r_b$ are erased:} This is similar to the case of erasing $r_b$. Instead of getting the wrong $r_a$, Bob gets $\bot$, but he executes the same (and so does Alice).
    \item\textbf{Only $r_b$ is flipped:} This is similar to the case of erasing $r_b$. Here, instead of getting~$\bot$, Alice gets the wrong~$r_b$, but she still executes the same lines.
    \item\textbf{$r_a$ is erased, $r_b$ is flipped:} Similar to the case of erasing $r_a,r_b$. Instead of getting $\bot$, Alice gets the wrong $r_b$, but she executes the same lines.
    \item\textbf{Both $r_a,r_b$ are flipped:} Alice raises $r_a$ (Line~\ref{line:alice:incRA}), and $r_a=r_b\mod2$. Due to the flip of $r_a$, Bob gets $r_a$ which fits to the expected, raises $r_b$ (Line~\ref{line:bob:increaseRB}) so $r_a \neq r_b \pmod 2$, inserts the message from Alice (Line~\ref{line:bob:insertAlice}) and sends Alice the next message (Line~\ref{line:bob:send}). Due to the flip of $r_b$, Alice gets $r_b$ which fits to the expected and inserts the received symbol (Line~\ref{line:alice:insertBob}). At the end of the iteration $r_a \neq r_b \pmod 2$, and both parties make the same progress.
\end{enumerate}
\qquad\tikz{
\draw (1,0) -- (\columnwidth,0);
\draw[fill=black] (1,0) circle (0.075);
\draw[fill=black] (\columnwidth,0) circle (0.075);
}
\begin{enumerate}[resume*,topsep=0pt,parsep=0pt]
    \item\textbf{No corruptions:} Alice raises $r_a$ (Line~\ref{line:alice:incRA}), and now $r_a=r_b\pmod2$. Bob gets~$r_a$ which does not fit to the expected, and sends Alice the message from his memory  (Line~\ref{line:bob:previous}). Thus, Alice gets a message with $r_a=r_b \mod 2$ and inserts the message from Bob (Line~\ref{line:alice:insertBob}). At the end of this iteration $r_a=r_b \pmod 2$ and only Alice has made progress.
    \item\textbf{Only $r_a$ is erased:} This is similar to the case of no corruptions. Now, instead of getting the wrong~$r_a$, Bob gets~$\bot$ but still performs in the same manner (and so does Alice).
    \item\textbf{Only $r_a$ is flipped:} Alice raises $r_a$ (Line~\ref{line:alice:incRA}) and now $r_a=r_b\pmod2$. Bob gets $r_a$ which fits to the expected, raises $r_b$ (Line~\ref{line:bob:increaseRB}) so $r_a \neq r_b \pmod 2$, and sends Alice the next message (Line~\ref{line:bob:send}). Alice gets $r_b$ which does not fit to the expected, decreases $r_a$ (Line~\ref{line:alice:decRA}) and deletes the last symbol (Line~\ref{line:alice:delete}). At the end of the iteration $r_a=r_b\pmod2$ (since Bob changed his $r_b$) and (only) Bob makes progress.
    \item\textbf{ $r_a$ is flipped,  $r_b$ is erased:}
    this is similar to the case of flipping~$r_a$. Instead of getting the wrong $r_b$, here Alice gets~$\bot$, but she executes the same code as in the previous case.
\end{enumerate}
\end{description}
Note that after cases (1)--(5) it holds that 
$r_a \neq r_b\pmod2$, thus the sequence does not end after this iteration. 
Further, in cases (1)--(4) both parties make no progress, yet in case~(5), both parties make the same progress. These cases correspond to part (2)(b) in the lemma's statement. In each one of these cases there is at least one corruption.

In cases (6)--(9), exactly one party makes progress and we end up having $r_a = r_b\pmod2$, thus, this is the last iteration of the sequence, corresponding to part (2)(c) in the lemma's statement.
We note that, as stated, in the two cases where there are no flips, cases (6)--(7), Alice is the party that makes progress. 

\smallskip
The last claim we have to prove is that when there are no flips in the sequence~$[i,j]$, then the progress of Bob at iteration~$i$ equals the progress of Alice at  iteration~$j$. 
To see that, we denote the parties transcripts at the beginning of iteration $i$ as $\T_a, \T_b$, and the parties inputs as $x, y$. At iteration~$i$ Alice sends $\sigma=\pi'(x\mid\T_a)$, which is not flipped by the claim condition, and not erased since otherwise Bob would not progress in iteration~$i$, in contradiction. Thus, Bob's progress at iteration~$i$ consists of $\sigma$ and $\rho=\pi'(y\mid\T_b \circ \sigma)$. 
Alice does not make progress in iterations $[i, j-1]$ at all, and Bob does not make progress in iterations $[i+1, j]$ at all. 
Thus, at the beginning of iteration~$j$ Alice's transcript is $\T_a$. Alice sends Bob $\sigma=\pi'(x\mid\T_a)$ (Line~\ref{line:alice:compute}), which is added to her transcript. 
Since Bob does not make progress in iteration~$j$, he sends Alice the message he sent in iteration $j-1$ (Line~\ref{line:bob:previous}), which is the message he sent in iteration $j-2$ and so on, which is~$\rho$. Therefore, Alice's progress in iteration~$j$ consists of $\sigma$ and $\rho$, and this equals to Bob's progress in iteration~$i$.
\end{proof}

We can now analyze a slightly more general case, where corruption can affect either the $m$ or the $r_a,r_b$ part of every message, but transmission that belongs to the~$m$ part can only be \emph{erased}. 
\begin{lemma}
\label{lem:sequence:manage:content:erasures}
    Assume that flips corrupts only transmissions where $r_a,r_b$ are communicated while erasures may corrupt any transmission.  Then, Lemma~\ref{lem:sequence:all} holds.
\end{lemma}

\begin{proof}
Assume the transmission of some~$(m,r)$. If the $m$ part is not erased, we are back in the premises of Lemma~\ref{lem:sequence:manage}. Otherwise, we argue that the parties behave identically to the case where $m$ is uncorrupted and $r$ is erased. Indeed, in both cases the receiver of the message executes the same lines (Line~\ref{line:bob:err} or Line~\ref{line:alice:delete}), which changes its transcript in the same manner (i.e., does not change at all for Bob, and deleting the last symbol for Alice). The state of a party is determined by the transcript, and since these are identical in both cases, the next message that the receiver of $(m, r)$ will send later on will be the same, etc.

Since erasing $m$ is equivalent to erasing $r$ (while keeping $m$ uncorrupted), we are again in the premises of Lemma~\ref{lem:sequence:manage}, which completes the proof.
\end{proof}

We can finally analyze the general case, where both flips and erasures can affect either the $m$ or the $r_a,r_b$ parts of each message.
\begin{proof}[Proof of Lemma~\ref{lem:sequence:all}]
Consider the transmission of some~$(m, r)$.
The only corruption cases which still require a proof (beyond the cases covered by Lemma~\ref{lem:sequence:manage:content:erasures}) are when the $m$~part is flipped, and the $r$ part is either flipped, erased or remain uncorrupted. 

When the $r$~part is erased, similar to the proof of Lemma~\ref{lem:sequence:manage:content:erasures}, the $m$ part is anyways disregarded, and the proof follows in a similar manner.

Now assume that $r$ is either flipped or uncorrupted. 
Note that the parties decide whether to  progress or discard the iteration 
depending only on the $r$ bits (assuming no erasures in the $m$ part, see lines~\ref{line:alice:if} and~\ref{line:bob:if}). Thus, the flips that occur in the $m$~part have no effect on the question of \emph{who} makes progress in some iteration, whether it is Alice, Bob, both of them or neither of them. 
Note that the parts of the statements in Lemma~\ref{lem:sequence:all} that deal with which party makes an insertion and whether there is progress or not, depend only on this question. Hence, the previous lemmas prove Lemma~\ref{lem:sequence:all} up to the part concerning whether the progress is correct or not, which depends on the value of the $m$ bits, and we complete now.

Let us analyze the progress made in some iteration~$k$, assuming there is at least one flip in at least one $m$-part in~$k$.

Cases (1) and (2)(b) follow vacuously, since there are flips in some $m$-part in iteration~$k$.
Case (2)(c) states that the progress of Alice in iteration $j$ equals to the progress of Bob in iteration $i$, if there are no flips in the sequence at all.

We conclude that all the statements concerning the \emph{content} of the progress made in $k$ only consider the case when there are no flips in the $m$-parts, and thus the correctness of Lemma \ref{lem:sequence:all} stems from Lemma \ref{lem:sequence:manage:content:erasures}. 
\end{proof}

Before we proceed, we present a helpful observation.
\begin{observation}
\label{obs:sequence-cover}
    Consider a sequence $[i, j]$, where $i>1$. Then, each iteration $1 \le k \le i$ is included within a sequence, whose last iteration precedes $i$.
\end{observation}
\begin{proof}
    The proof is by induction on $k$. At the beginning of iteration 1, $r_a=r_b \mod 2$. The first iteration in which $r_a=r_b\mod 2$ at its end is the last iteration of the sequence which includes iteration 1. Iteration $i-1$ fulfills this condition; thus, iteration 1 is included within a sequence, with the last iteration preceding $i$.
    
    Assume iteration $k<i$ belongs to a sequence. If $r_a=r_b \mod 2$ at the end of $k$, then $r_a=r_b \mod 2$ at the beginning of $k+1$. Hence, iteration $k+1$ belongs to a sequence (similar to iteration~1). Otherwise, $k$ is not the last iteration of a sequence, and $k+1$ is included within the same sequence as $k$.
\end{proof}

\subsection{Frames and Segments}
The notion of a sequence is instrumental in understanding the behavior of the scheme~$\Pi$ (Algorithms~\ref{alg:schemeALICE} and~\ref{alg:schemeBOB}), and for relating its progress with the amount of corruptions it suffers. 
As mentioned above, our approach 
towards proving the main theorem is showing a reduction between $\Pi$ in the UPEF model and $\pi'$ in the insertion-deletion model.
Towards that end, we would like to be able to connect any partial execution (say, a sequence) of~$\Pi$ to an equivalent partial execution of~$\pi'$, and connect the noise in one model to the noise in the other.
Unfortunately, the notion of a sequence is insufficient for this goal. It can happen that a single transmission in~$\pi'$ is explained by several consecutive sequences in~$\Pi$.  
To see that, consider the sequence where Bob inserts twice (say, $a_1,b_1$ and $a_2,b_2$) but Alice makes no progress, followed by the sequence where Bob inserts at the first iteration (say, $a_3,b_3$) and Alice inserts at the second iteration (say, $a_4, b_4$). 
This execution translates to 4 rounds of~$\pi'$: at the first round Alice sends $a_4$ , which is substituted by the noise to~$a_1$ (note that Alice sends the same message in each of the four iterations, since at each iteration she thinks that the previous iteration had an error). 
Then, Bob suffers 2 out-of-sync corruptions, which consist of $b_1, a_2, b_2, a_3$. At the last round Bob replies with $b_3$, which is substituted by the noise to~$b_4$.
The above discussion hints that sequences are somewhat insufficient and that we need a unit of larger scale in~$\Pi$ in order to capture atomic operations of~$\pi'$. Towards this goal we define a new structure we call \emph{a frame}.

\begin{definition}[A good sequence]
    A sequence $[i, j]$ is \emph{good}
    if Alice makes progress at iteration~$j$. Otherwise, the sequence is \emph{bad}.
\end{definition}

\begin{definition}[A frame]
\label{def:frame}
    Assume some sequence $s_j = [j,j']$ is good, and let $s_i=[i,i']$ be the first bad sequence satisfying: (1) all the sequences between iterations $i$ and $j-1$ are bad and (2) $s_i$ has at least one iteration with progress. In particular, note that iteration $i$ has progress. If no such $s_i$ exists, set $s_i=s_j$.
    Then, the consecutive series of sequences $s_i,\ldots,s_j$ is called \emph{a frame}.
\end{definition}

We are going to set some definitions in order to establish a reduction between a frame and some partial execution of~$\pi'$. From now on, we use the notation of $s_i=[i,i']$ for sequences.

\begin{definition}[\Good transcripts]
The (possibly partial) transcripts $\T_a, \T_b$ are called \emph{\good} if there exists some execution of $\pi'$ that generates $\T_a,\T_b$,
in which Alice receives a message at the last step (either from Bob or due to an out-of-sync noise).
\end{definition}

\begin{definition}[Matching execution]
Let $\T_a, \T_b$ and $\T_a', \T_b'$ be the transcripts at the beginning and ending, respectively, of iterations~$[i,j]$, and assume $\T_a, \T_b$ are \good.
We say that $\pi'$ has a \emph{matching execution to~$[i,j]$} if
there exists a noise pattern (in the insertion-deletion model) such that executing $\pi'$ with that noise pattern starting with transcripts $\T_a, \T_b$,  yields the transcripts $\T_a', \T_b'$.
\end{definition}

\begin{theorem}
    \label{thm:frame}
    Assume an execution of $\Pi$ on some inputs $x, y$, and let $s_i,\ldots,s_j$ be a frame during which $f$~bit flips and $d$~erasures have occurred. 
    Let $\T_a, \T_b$ be Alice's and Bob's respective variables at the beginning of iteration $i$, and denote by $\T_a', \T_b'$ their respective variables at the end of iteration~$j'$. 
    If $\T_a, \T_b$ are \good transcripts,
    then there exist a matching execution of $\pi'$ for $s_i,\ldots,s_j$, suffering at most $2f$ edit-corruptions. It holds that  $\T'_a, \T'_b$ are \good transcripts.
\end{theorem}

The proof of Theorem~\ref{thm:frame} spans the rest of this section.
We do this by dividing the frame into small parts we call \emph{segments}, and reducing each segment to an equivalent execution of~$\pi'$ (Lemmas \ref{lem:reduction:not-last}--\ref{lem:reduction:last}). Finally, we complete the proof of Theorem~\ref{thm:frame} by `stitching' together all the segments (and the iterations with no progress that might reside in between) that are contained in a frame.

\subsubsection{Frame Observations}
First, we make two observations on frames, which will be helpful in their further analysis. The first one tells about the structure of sequences in a frame, stemming from the fact that all the sequences in a frame, except from the last one, are bad. The second observation tells us that given some frame, all the iterations with progress before that frame must belong to frames. That observation is crucial for using frames to analyze executions of $\Pi$, since it indicates that between frames there are only iterations without progress, which have limited effect on the simulation.

\begin{observation}
\label{obs:frame:not-last-sequence}
    Let $s_i, \ldots, s_j$ be a frame, and assume $s_i \neq s_j$. Then, each sequence between iterations $[i, j-1]$ either (1) has multiple iterations and Bob inserts at the sequence's last iteration, or (2) has single iteration where no party makes progress. In particular, $s_i$ satisfies (1) since it has progress (by the frame definition).
\end{observation}
\begin{proof}
    Let $s=[k,k']$ be a sequence between iterations $[i, j-1]$. By the definition of a frame, at iteration $k'$ Alice does not make progress, or otherwise the frame ends there. By Lemma~\ref{lem:sequence:all}, either $k=k'$ or $k \neq k'$. When $k=k'$, then by Lemma~\ref{lem:sequence:all}(1) either both parties make progress in $k'$, or no party makes progress in $k'$. Since Alice does not make progress in $k'$, then no party makes progress in $k'$, and case (2) is satisfied.

    When $k \neq k'$, then by Lemma~\ref{lem:sequence:all}(2)(c) either Alice or Bob inserts in $k'$. Since Alice does not makes progress in $k'$, then Bob inserts in $k'$ and the proof is complete.
\end{proof}

\begin{observation}
    \label{obs:frame-complete}
    Let $s_i, \ldots, s_j$ be a frame, and let $k < i$ be an iteration with progress. Then, there exists some frame that includes iteration~$k$.
\end{observation}

\begin{proof}
    By Observation \ref{obs:sequence-cover}, iteration $k$ is included within a sequence. Since $k$ is an iteration with progress, the sequence contains iteration $k$ is a sequence with progress, whose last iteration $l$ satisfies $l < i$. Denote the last sequence satisfying these two properties by $s$, and its last iteration by~$k'$. 
    Due to Lemma~\ref{lem:sequence:all}(1), (2)(c), in the last iteration of a sequence with progress, at least one party makes progress. 
    
    Assume towards contradiction that only Bob makes progress in~$k'$. Thus, $s$ is a bad sequence with progress. We argue that all the sequences from $s$ to the sequence ending at $j-1$ are bad. The sequences after $s$ through $s_i$, if any, are bad because they have no progress. The sequences between $[i, j-1]$ are bad, due to frame definition. Combining these observations, we deduce that all the sequences from $s$ to the sequence ending at $j-1$ are bad. 
    
    Thus, $s$ is a prior sequence to $s$ that satisfies (1) all the sequences between $s$ and the sequence ending at $j-1$ are bad and (2) $s$ has at least one iteration with progress. By frame definition, $s_i$ is supposed to be the first sequence which satisfies these conditions, which leads to a contradiction.
     
    Therefore, Alice makes progress in~$k'$. Thus, $s$ is good, which makes $s$ the last sequence of some frame $F$, which comes before the frame $s_i,\ldots,s_j$. 
    If $k$ is included within the frame~$F$, then the proof is complete. Otherwise, $k$ is prior to the first iteration of~$F$, and we can apply the same argument inductively on $k$ and~$F$.
\end{proof}

\subsubsection{Segments}
\label{apdx:segments}
As previously stated, the proof of Theorem \ref{thm:frame} is achieved by dividing the frame into smaller sections, which are referred to as ``segments".

\begin{definition}[A segment]
\label{def:segment}
     Let $k'$ be an iteration where Alice makes progress, and let~$k<k'$ be the first iteration such that (1) Alice does not make progress in any of the iterations in $[k,k')$ and (2) Bob makes progress in $k$. If no such $k$ exists, set $k=k'$. 
    The consecutive series of iterations $[k, k']$ is called \emph{a segment}.
\end{definition}

The objective of this section is to reduce segments into executions of $\pi'$. We start with a straightforward observation on the structure of segments, which we explicitly define in order to ease the analysis of a segment.
\begin{observation}
    \label{obs:segment:types}
    Let $q=[k,k']$ be a segment. Exactly one of the following holds:
    \begin{enumerate}
        \item Alice inserts at~$k'$, and $k=k'$.
        \item Alice and Bob make progress at~$k'$, and $k=k'$.
        \item Alice inserts at~$k'$, and Bob inserts at iterations $P = \{k_1, \ldots, k_{|P|}\}$ where $k_1=k$, $P \neq \emptyset$ and $k_i<k'$ for all $i$.
        \item  Alice and Bob make progress at~$k'$, and Bob inserts at iterations $P = \{k_1, \ldots, k_{|P|}\}$ where $k_1=k$, $P \neq \emptyset$ and $k_i<k'$ for all $i$.
    \end{enumerate}
\end{observation}

\begin{proof}
    By Definition~\ref{def:segment}, Alice makes progress at iteration $k'$, and that is the only iteration in~$q$ where Alice makes progress. Bob either makes or does not make a progress at~$k'$, and there is or there is not an iteration before $k'$ where he makes progress, which leads to the four cases. Note that by segment definition, in iteration $k$ at least one party makes progress; thus, in cases (3) and (4), Bob makes progress in $k$, so $k_1=k$.
\end{proof}

We are going to analyze segments within frames, and thus we state the following lemmas about segments in frames. Recall the notation of $s_i=[i,i']$ for sequences.
\begin{lemma}
\label{lem:segment:not-cross}
    Let $q=[k,k']$ be a segment, where $k'$ is in some frame $s_i, \ldots, s_j$. Then, $k$ is also in $s_i, \ldots, s_j$.
\end{lemma}

\begin{proof}
    If $q$ is of type (1) or (2) of \Cref{obs:segment:types}, then $k=k'$, and the proof is immediate. If $q$ is of type (3) or (4), Bob makes progress in $k$ and Alice makes no progress in iterations $[k,k')$. Assume towards contradiction that $k < i$. Then, $k$ belongs to another frame~$F$, which is prior to $s_i \ldots s_j$, due to Observation~\ref{obs:frame-complete}.
    By frame definition, Alice makes progress at the last iteration of~$F$, say, iteration~$k''$. However,  $k \le k'' < k'$, which contradicts the fact that Alice does not make progress in $[k,k')$.
\end{proof}

\begin{lemma}
\label{lem:frame:segments-cover}
    Let $s_i, \ldots, s_j$ be a frame. Then each iteration with progress in $[i,j']$ is included within a segment. 
    In particular, iterations $i, j'$ are included in the first and last segment within the frame, respectively. It should be noted that these segments may, in fact, be identical.
\end{lemma}

\begin{proof}
    We prove the claim by scanning the iterations from $j'$ backwards to~$i$. Each time we hit an iteration $k'$ where Alice makes progress. This defines a segment $q_k=[k,k']$, that covers more iterations that might contain progress. Then, we find a prior segment $q_l=[l,l']$, and show that it either ends at $l'=k-1$, or that between $l'+1$ and  $k-1$ there was no progress. We keep doing that until we reach iteration~$i$, thus covering all the iterations with progress in the frame.
    
    By frame definition, Alice makes progress in $j'$, the last iteration of the frame. Using segment's definition, $j'$ belongs to some segment, say $q_k=[k,k'=j']$.
    By Lemma~\ref{lem:segment:not-cross}, segments do not cross frames, thus $k\ge i$. If $k=i$, then the proof is complete. 
    
    Otherwise, $k > i$. Consider iteration $k-1$.
    By segment definition, there are two possibilities here:
    either Alice makes progress, or no one makes progress. Indeed, the case where only Bob makes progress is impossible, since it would mean that $q_k$ could extend to the iteration $k-1$.
    
    Let us consider the two cases in turn. 
    If Alice makes progress in $k-1$, then we can define a new segment $q_l=[l,l']$ with $l'=k-1$, and repeat the argument for this segment.

    We now turn to the other case, where Alice 
    does not make progress in~$k-1$. Yet, we argue that there is some iteration $l' \in [i,k-1)$ in which Alice makes progress, while there is no progress during $[l'+1,k-1]$. First, we argue that Alice must make progress in $[i,k-1)$. Indeed, by definition, iteration~$i$ has progress. If only Bob makes progress in iteration~$i$, and Alice never makes progress in $[i,k-1)$, then $q_k$ (which is the latest segment we defined so far) could extend to iteration $i$, in contradiction. 

    We have established that Alice makes progress in $[i,k-1)$. Then, denote the last iteration before $k-1$ where Alice makes progress by $l'$.
    We argue that no party makes progress during $[l'+1, k-1]$. Alice does not make progress by definition of $l'$; if Bob makes progress during $[l'+1, k-1]$, then by segment definition $q_k$ could extend to the last iteration before $k-1$ where Bob makes progress, in contradiction.
    Then, since Alice makes progress at~$l'$, we can define a segment $q_l=[l, l']$, and repeat the argument above until we reach $l=i$. Every time we repeat that argument, the first iteration of the created segment is at least $i$, due to Lemma \ref{lem:segment:not-cross}.
\end{proof}

In the rest of this section, we use the notation of $P$ to denote the set of iterations in which Bob inserts during segments of types (3) and (4) according to Observation~\ref{obs:segment:types}. The following lemma shows that certain types of segments can appear in a frame only in certain positions.

\begin{lemma}
    \label{lem:segment:props}
    Let $q=[k,k']$ be a segment within the frame  $s_i, \ldots, s_j$.
    \begin{enumerate}
        \item If $q$ is of type (1), it is either the first or the last segment within a frame containing multiple segments. In particular, $q$ cannot be a middle segment, nor the only segment. 
        Further, the sequence~$s_i$ contains multiple iterations.
        \item If $q$ is of type~(2), then $q$ is either the only segment within a frame containing a single sequence~$s_i$ with a single iteration, $i=i'$, or $q$ is a middle-segment within a frame containing multiple segments.
    \end{enumerate}
\end{lemma}

\begin{proof}
    \begin{enumerate}
        \item Assume towards contradiction that $q=[k,k']$ is of type (1) and is neither the first nor the last segment within the frame. By Observation \ref{obs:sequence-cover}, iteration $k'$ belongs to a sequence. Since Alice inserts during $k'$, by Lemma~\ref{lem:sequence:all}(2)(a),(c) iteration $k'$ is either the first or the last iteration of some sequence $s$ in the frame. If $k'$ is the last iteration of~$s$, then by definition the frame must end there, so $s = s_j$. Then, $q$ is the last segment within the frame, in contradiction. 
        
        Thus, $k'$ must be the \emph{first} iteration of~$s$. We argue that $s \ne s_i$; otherwise, $i=k'$, and $q$ is the first segment within the frame, in contradiction. 
        Since $s_i$ has progress by frame definition, it follows that during the frame there is a sequence with progress before $s$. Denote the last such sequence by $s'$, and its last iteration by $k''$. By Observation~\ref{obs:frame:not-last-sequence} Bob inserts at $k''$. In addition, note that $k''$ is the last iteration with progress before $k'$, by definition of $s'$.
        Thus, $k''$ belongs to $q$, by segment definition. It follows that $q$ is no longer of type (1), in contradiction.

        We proceed to show that $q=[k ,k']$ cannot be the single segment within the frame. In a similar way to the previous case, iteration $k'$ is either the first or the last iteration of some sequence $s$ in the frame. By Lemma \ref{lem:sequence:all}(2)(a),(c), there is another iteration $l$ in $s$ where some party inserts (if $k'$ is the first iteration of $s$, then $l$ is its last iteration, and vice versa). By Lemma \ref{lem:frame:segments-cover}, $l$ belongs to some segment, which cannot be $q$ since $q$ is of type (1). This leads to a contradiction with the singularity of $q$ in the frame.
        
        We continue to the last part of the claim, and show that $s_i$ contains multiple iterations in a frame that contains a segment of type (1). 
        Assume towards contradiction that $s_i$ is a single-iteration sequence. Then, by frame definition, iteration $i=i'$ has progress, and by Lemma~\ref{lem:sequence:all}(1) both parties make progress in $i=i'$. Therefore $s_i$ is a good sequence, and the frame ends in $i'$, by frame definition. Thus, both parties make progress in the single iteration of the frame, and the only segment within the frame must be of type~(2), in contradiction.

        \item 
        We divide the proof into two cases: where $q$ is the only segment within the frame, and where there are additional segment(s) in the frame.
        
        Let $q=[k,k']$ be the single segment within the frame. First, we show that there is only one sequence in the frame.
        Suppose towards contradiction that the frame contains more than one sequence, such that $s_i\ne s_j$.
        By Observation~\ref{obs:frame:not-last-sequence}, $s_i$ has multiple iterations, and Bob inserts in $i'$. Since $i'$ has progress, then by Lemma \ref{lem:frame:segments-cover} it belongs to the single segment within the frame, which is $q$. Thus, $q$ cannot be of type (2), in contradiction.
        
        We proceed to show that the single sequence $s_i=s_j$ in the frame has only one iteration. By frame definition, $j'$ has progress. Thus, by Lemma \ref{lem:frame:segments-cover}, $j'$ is included within the single segment within the frame, $q=[k,k']$, and $k'=j'$. 
        Since $q$ is of type (2), both parties make progress in $j'$, and Lemma~\ref{lem:sequence:all}(1) suggests that $j=j'$. Thus, $j=j'=k=k'$ and the proof of the case where $q$ is the only segment within the frame is complete.

        Next, assume that there are multiple segments in the frame. We show that $q$ is a middle segment.
        Assume towards contradiction that $q$ is the first segment.
        By frame definition, iteration $i$ has progress, and by Lemma~\ref{lem:frame:segments-cover} it belongs to the first segment within the frame, which is~$q$.
        Since $q$ has a single iteration (being type~(2)), $i=k=k'$, and both parties make progress at iteration~$i$.
        By Lemma~\ref{lem:sequence:all}(1), $s_i$ has only one iteration. Further, it is a good sequence by definition.
        Thus, the frame ends at $s_i$, so $s_i = s_j$ and $q$ is the single segment within the frame, in contradiction.
        
        We continue to the case in which $q$ is the last segment within the frame.
        In the same manner to the previous case, and by frame definition, iteration $j'$ has progress, and belongs to the last segment within the frame, which is $q$.
        Since $q$ has a single iteration (being type~(2)), $j'=k=k'$. By Lemma~\ref{lem:sequence:all}(1), $s_j$ has only one iteration, hence $k=k'=j=j'$.
        Nevertheless, reaching contradiction is somewhat more challenging than in the preceding case. By assumption, there are multiple segments in the frame. By segment definition, a segment has at least one iteration with progress. By Observation \ref{obs:sequence-cover}, each iteration in a frame belongs to a sequence. Thus, there exists a sequence with progress in the frame before $s_j$. Let $s_l=[l,l']$ be the last such sequence.
        Observation~\ref{obs:frame:not-last-sequence} tells us that Bob inserts at iteration~$l'$. Since there is no progress during $[l'+1, j'-1]$, $q$ includes $l'$ by segment definition, in contradiction to the assumption that $q$ is of type~(2).
    \end{enumerate}
\end{proof}

Before we proceed to analyze the segments of types (3) and (4), we present the following lemma, which concerns the set $P$ of Bob's insertions.
\begin{lemma}
\label{lem:segment:p-set}
    Let $q=[k,k']$ be a segment of type (3) or (4) in the frame $s_i, \ldots, s_j$. Then, each $k_l \in P$ is either the first or the last iteration of a sequence. Furthermore, let $[l, l']$ be a sequence with progress in the given frame. If $k \leq l < k'$, then $l \in P$; and similarly, if $k \leq l' < k'$, then $l' \in P$.
\end{lemma}
\begin{proof}
    By Observation~\ref{obs:segment:types}, Bob inserts in each $k_l \in P$. By Observation \ref{obs:sequence-cover}, each iteration of the frame is included within a sequence. In particular, each iteration of $P$ is included within a sequence. By Lemma~\ref{lem:sequence:all}(2)(a),(c), an insertion may only occur at either the first or the last iteration of a sequence. Thus, each $k_l \in P$ is the first or the last iteration of a sequence. 
    
    Let $s=[l, l']$ be a sequence with progress. We prove that if $k \leq l < k'$, then $l \in P$; the case of $l'$ is left for the reader. Assume towards contradiction that $l=l'$. Since $s$ has progress, it follows from Lemma~\ref{lem:sequence:all}(1) that both parties have made progress in $l$, with Alice in particular. Thus, $l = k'$, in contradiction to the requirement $l < k'$. Thus, $l < l'$. By Lemma~\ref{lem:sequence:all}(2)(a),(c), iterations $l,l'$ end up with an insertion. If Alice inserts at $l$, then $l = k'$ in contradiction. Thus, Bob inserts at $l$. By definitions of types (3) and (4), $P$ contains all the iterations in $q$ where Bob inserts, hence $l \in P$.
\end{proof}

\begin{lemma}
    \label{lem:segment:flips}
    Let $q=[k,k']$ be a segment within the frame  $s_i, \ldots, s_j$.
    \begin{enumerate}
        \item Assume $q$ is of type~(3). If $q$ is the last segment within the frame, then at least $(|P|-1)/2$ $r$-bits are flipped during $q$. Otherwise, at least $(|P|+2)/2$ $r$-bits are flipped during $q$.
        \item Assume $q$ is of type (4). If $q$ is the last segment within the frame, then at least $|P|/2$ $r$-bits are flipped during $q$. Otherwise at least $(|P|+3)/2$ $r$-bits are flipped during $q$.
    \end{enumerate}
\end{lemma}

\begin{proof}
    We begin by presenting the principal arguments supporting every one of the claims, as they are similar. Each segment of type (3) or (4) is comprised of a series of iterations, during which Bob inserts, and a final iteration during which Alice makes progress, sometimes accompanied by Bob. 
    The principal component of the asserted lower bound on the count of flipped $r$-bits comprises the iterations during which Bob inserts. The precise bound, provided for each case, is dependent upon the characteristics of the final iteration of $q$, in terms of its position within the frame.
    \begin{enumerate}[topsep=0pt,parsep=0pt]
        \item 
        In this section we assume that $q$ is of type (3). Let us first discuss the case where $q$ is the last segment within the frame. By Lemma~\ref{lem:frame:segments-cover}, iteration $j'$ belongs to $q$. By frame definition, Alice makes progress in $j'$, hence $k'=j'$. Since $q$ is of type (3), Alice inserts at $k'$. Thus, by Lemma \ref{lem:sequence:all}(2)(c), $s_j$ has multiple iterations, i.e. $j \ne j'$. By Lemma~\ref{lem:sequence:all}(2)(a) either Alice or Bob inserts at iteration $j$. Assume towards contradiction that Alice inserts at $j$. By Lemma \ref{lem:sequence:all}(2)(b), there is no iteration in the interval $(j, j')$ where some party inserts. In particular, at no point during these iterations Bob inserts. Consequently, $q=[j',j']$, hence it is a segment of type (1), in contradiction. Hence, Bob inserts at $j$, and $k_{|P|} = j$.
        By combining with Lemma~\ref{lem:segment:p-set}, it can be shown that $\lfloor{|P|/2}\rfloor$ of the iterations in $P$ are the ending iterations of sequences, while the remaining are the starting iterations of sequences. 
        By Lemma~\ref{lem:sequence:all}(2)(c), at least one $r$-bit is flipped during an insertion of Bob which is the last iteration of a sequence. Hence, at least $\lfloor{|P|/2}\rfloor$ $r$-bits are flipped between $k_1$ and $k_{|P|}$, which are at least $(|P|-1)/2$ $r$-bits.
        
        We proceed to the case in which $q$
        is not the last segment within the frame. Since $q$ is of type (3), Alice inserts at iteration $k'$. By Observation \ref{obs:sequence-cover}, iteration $k'$ belongs to a sequence $s=[l,l']$. It follows from Lemma~\ref{lem:sequence:all}(2)(a),(c) that $l \ne l'$, and either $k'=l$ or $k'=l'$. Assume towards contradiction that $k'=l'$. By frame definition, $s = s_j$, and $q$ is the last segment within the frame, in contradiction. Thus, $k'=l$. We proceed to demonstrate that $k_{|P|}$ is the last iteration of a sequence. First, we show that $s \neq s_i$; otherwise, $k'=l=i$, and it follows from Lemma \ref{lem:segment:not-cross} that $k=i$ such that $q$ is of type (1), in contradiction. By frame definition, the sequence $s_i$ has progress, and it follows that there is a sequence with progress before $s$. Let $s_1=[l_1, l_1']$ be the last such sequence. By Observation~\ref{obs:frame:not-last-sequence}, Bob inserts at $l_1'$. Hence, $l_1'$ is the last iteration with progress before $k'$, and $l_1' = k_{|P|}$. Therefore, $k_{|P|}$ is the last iteration of a sequence. In the same manner to the case where $q$ is the last segment of the frame, we can see that at least $\lceil{|P|/2}\rceil$ $r$-bits are flipped between $k_1$ and $k_{|P|}$. Since Alice inserts in $k'=l$, it follows from Lemma~\ref{lem:sequence:all}(2)(a) that at least one $r$-bit is flipped during $k'$. For conclusion, at least $|P|/2+1=(|P|+2)/2$ $r$-bits are flipped during $q$.

        \item %
        In this section we assume that $q$ is of type (4). First, we establish that $q$ is the last segment within the frame if and only if iteration $k'$ is a single iteration in a sequence.
        
        If $q$ is the last segment within the frame, then it follows from Lemma~\ref{lem:frame:segments-cover} that $k' = j'$. By definition of type (4), both parties make progress in $k'$. Thus, by Lemma~\ref{lem:sequence:all}(1) , $j = j' = k'$, and iteration $k'$ is a single iteration in a sequence. 
        We now turn to the second part of the claim, assuming that iteration $k'$ is a single iteration in a sequence $s$. By segment definition, Alice makes progress in $k'$. Since $k'$ is the last (and only) iteration of $s$, it follows from frame definition that $s=s_j$, hence $q$ is the last segment within the frame. 
        
        We now proceed to the calculation of the lower bound on the number of flips of $r$-bits during $q$. First, assume $q$ is not the last segment within the frame. Consequently, $k'$ is not a single iteration in a sequence, and it follows from Lemma~\ref{lem:sequence:all}(2)(b) that $k'$ is neither the first nor the last iteration of the sequence $s=[l,l']$ to which it belongs. By Lemma~\ref{lem:sequence:all}(2)(b), the two $r$-bits sent during iteration $k'$ are flipped. Since $q$ is of type (4), it can be easily shown that Bob inserts in $l$. Thus, $k_{|P|}=l$.
        In a similar argument to the case of segment of type (3), at least $(|P|-1)/2$ $r$-bits are flipped between $k_1$ and $k_{|P|}$, and at least $(|P|-1)/2+2=(|P|+3)/2$ $r$-bits are flipped in $q$.
        
        If $q$ is the last segment within the frame, then $k'$ is a single iteration in a sequence $s$. Since there is progress during the frame before $s$ (similar to the case in which $q$ is the last segment and is of type (3)), there is a sequence with progress before $s$. Let $s'=[l_1, l_1']$ be the last such sequence. By Observation~\ref{obs:frame:not-last-sequence}, Bob inserts at $l_1'$. Since there is no progress between $l_1'$ and $k'$, it follows that $k_{|P|} = l_1'$. A parallel argument to the that presented for segments of type (3) reveals that at least $\lceil{|P|/2}\rceil$ $r$-bits are flipped between $k_1$ and $k_{|P|}$, which are at least $|P|/2$ $r$-bits.
    \end{enumerate}
\end{proof}

The following lemmas describe the reduction of segments in $\Pi$ into executions of $\pi'$. We distinguish between the last segment within a frame and the other segments.

\begin{lemma}
    \label{lem:reduction:not-last}
    Consider an execution of $\Pi$ on inputs $x,y$ in which $s_i,\ldots,s_j$ is a frame, and  $q=[k,k']$ is a segment which is not the last segment in that frame (in particular, assume that there are multiple segments in the frame). Assume that $f$ flips occurred during $q$. 
    
    Let $\T_a, \T_b$ be Alice's and Bob's respective variables at the beginning of iteration~$k$, and denote by $\T_a', \T_b'$ their respective variables at the end of iteration $k'$. 
    If $\T_a, \T_b$ are \good transcripts, then there exists a matching execution of $\pi'$ for $q$,
    with $c$ edit-corruptions, such that $2f - c > 0$. It holds that $\T'_a, \T'_b$ are \good.
\end{lemma}

\begin{proof}
    The proof is divided according to the segment types, as defined in Observation \ref{obs:segment:types}. For each type we present an execution of~$\pi'$ starting at the state $\T_a, \T_b$ and ends with $\T'_a, \T'_b$. We state explicitly the noise pattern during the presented execution. In particular, we show that $\T_a', \T_b'$ are \good.

    \begin{enumerate}[topsep=0pt,parsep=0pt,label={type (\arabic*)},left=0pt]
        \item We first show that $f \ge 1$. By Lemma~\ref{lem:segment:props}, $s_i$ has multiple iterations. Thus, it follows from Lemma~\ref{lem:sequence:all}(2)(a) that iteration $i$ ends up with an insertion. We argue that Alice is the inserting party at $i$; otherwise, by Lemma \ref{lem:frame:segments-cover}, iteration $i$ belongs to the first segment in the frame, and by Lemma \ref{lem:segment:props}, this segment is $q$. It implies that $q$ is of type (3) or (4), which is a contradiction. Due to Lemma~\ref{lem:sequence:all}(2)(a), at least one flip occurs during iteration~$i$, hence $f\ge 1$. 
        
        Denote by $\sigma, \rho'$ the two messages Alice adds to her transcript at iteration~$k'$, with $\sigma=\pi'(x\mid\T_a)$. Thus, $\T'_a= \T_a \circ \sigma \circ \rho'$ and $\T'_b=\T_b$ at the end of~$q$.
        We describe the matching execution of~$\pi'$ starting with $\T_a,\T_b$. Since $\T_a,\T_b$ are \good, Alice sends the next message, which is $\sigma=\pi'(x\mid\T_a)$. The noise deletes this message and injects the message $\rho'$ towards Alice; there are $c=1$ edit corruptions in this execution. The parties now hold the same $\T'_a,\T'_b$ as at the end of~$q$ in the execution of~$\Pi$.
        We conclude that $c=1$ and $f \geq 1$, hence $2f-c > 0$. The above $\pi'$ execution also implies that $\T_a', \T_b'$ are \good. 

        \item Let $s_l=[l, l']$ be the sequence containing $k'$. We argue that $l < k' < l'$. Since $q$ is a segment of type (2), both parties make progress in $k'$. Lemma \ref{lem:sequence:all}(1), (2)(b) suggests that either $l=k'=l'$ or $l < k' < l'$. If $l=k'=l'$, then $s_l$ is good, and $s_l=s_j$ by frame definition. Consequently, $q$ is the last segment in the frame, in contradiction.
        
        Consider iteration~$k'$ in~$\Pi$. Denote the $m$-part sent by Alice during $k'$ by $\sigma=\pi'(x\mid\T_a)$, and the message Bob receives by $\sigma'$. Then, Bob must send the $m$-part $\rho=\pi'(y\mid\T_b\circ\sigma')$ to Alice, who receives, say,~$\rho'$. At the end of iteration~$k'$ we have
        $\T'_a=\T_a \circ \sigma \circ \rho'$ and $\T'_b = \T_b \circ \sigma' \circ \rho$. Since $l < k' < l'$, it follows from Lemma~\ref{lem:sequence:all}(2)(b) that $r_a, r_b$ are flipped in $k'$, hence $f\ge 2$.
        
        We outline the matching execution of $\pi'$. Starting with the state $\T_a,\T_b$, the execution includes two rounds of messages. First, Alice sends Bob $\sigma$ and he receives (the potentially corrupted) $\sigma'$. Then Bob sends $\rho$ and Alice receives (the potentially corrupted) $\rho'$. The amount of edit-corruptions~$c$ in this execution, $0\le c \le 2$, counts the events $\sigma \ne \sigma'$ and  $\rho \ne \rho'$. Since $f\ge 2$ and $c\le 2$, we conclude that $2f - c > 0$. Similar to the proof of type~(1) above, $\T_a', \T_b'$ are \good.

        \item For each iteration $l\in[k,k']$, denote by $\sigma_l$ the $m$-part of the message Alice sends at iteration~$l$, and by $\sigma'_l$ the (potentially corrupted) $m$-part received by Bob. Similarly define $\rho_l$ and $\rho'_l$ as the messages Bob sends and Alice receives at iteration~$l$. Recall our notation of $P=\{k_1,\ldots, k_{|P|}\}$ for the iterations where Bob inserts during $q$ (Alice inserts only at~$k'$). It is immediate that at the end of iteration~$k'$ the parties hold
        \begin{align}
            \T_a' = \T_a \circ (\sigma_{k'} \circ \rho'_{k'}), \quad\text{ and }\quad
            \T_b' = \T_b \circ \bigcirc_{l=1, \ldots, |P|} (\sigma'_{k_l} \circ\rho_{k_l}).
        \end{align}
        We construct a matching execution of~$\pi'$, starting with $\T_a,\T_b$. 
        Since $\T_a,\T_b$ are \good, the execution starts with Alice sending Bob $\pi'(x \mid \T_a)=\sigma_{k_1}$.
        We set the noise such that Bob receives~$\sigma'_{k_1}$. 
        Then, for each $l \in \{1, \ldots, |P|-1\}$, Bob sends $\pi'(y\mid \T_b \circ \sigma'_{k_1} \circ \left.\bigcirc_{l'=2,\ldots,l}(\rho_{k_{l'-1}}\circ \sigma'_{k_{l'}})\right.)$. Clearly, this equals to $\rho_{k_l}$. We set the noise pattern such that Bob receives $\sigma'_{k_{l+1}}$ due to an out-of-sync error. 
        Finally, Bob sends $\rho_{k_{|P|}}=\rho_{k'}= \pi'(y \mid \T_b \circ \left.\bigcirc_{l=1, \ldots, |P|-1} (\sigma'_{k_l} \circ\rho_{k_l})\right. \circ \sigma'_{k_{|P|}})$, and Alice receives $\rho'_{k'}$. Since Alice makes no progress during the segment before iteration $k'$, $\sigma_{k_1} = \sigma_{k'}$. Therefore, the execution ends with the specified $\T'_a, \T'_b$, which are \good. Let $\delta$ be the count of the events  $\sigma_{k_1}\ne \sigma'_{k_1}$, $\rho_{k'} \ne \rho'_{k'}$. Note that $0 \leq \delta \leq 2$. Consequently, there are $\delta$ substitutions and $|P|-1$ out-of-sync errors, so $c = |P|-1+\delta \le |P|+1$. By Lemma~\ref{lem:segment:flips}, $f \geq (|P|+2)/2$ and thus $2f - c > 0$.

        \item Using the same notations of $\sigma_l, \sigma'_l, \rho_l, \rho'_l$ as in the proof of type~(3), we can write
        \begin{align}
        \T_a' = \T_a \circ (\sigma_{k'} \circ \rho'_{k'}), \quad\text{and}\quad \T_b' = \T_b \circ \bigcirc_{l=1, \ldots, |P|} (\sigma'_{k_l} \circ\rho_{k_l}) \circ (\sigma'_{k'} \circ \rho_{k'}).
        \end{align}
        We construct an execution of $\pi'$ that is analogous to the execution of case (3) above, with the exception that when Bob sends~$\rho_{k_{|P|}}$, this message does not reach Alice; instead, an out-of-sync noise causes Bob to receive~$\sigma'_{k'}$. 
        Then, Bob sends $\rho_{k'}$ and Alice receives $\rho'_{k'}$, so $\T_a', \T_b'$ are \good. In accordance with the same notation of $\delta$ as in case (3), and since there are $|P|$ out-of-sync errors, $c=|P|+\delta \le |P|+2$. By Lemma~\ref{lem:segment:flips}, $f \geq (|P|+3)/2$, and $2f-c>0$.
    \end{enumerate}
\end{proof}

\begin{lemma}
    \label{lem:reduction:last}
    Consider an execution of $\Pi$ on inputs $x,y$ in which $s_i,\ldots,s_j$ is a frame, and  $q=[k,k']$ is the last segment in that frame. Assume that $f$ flips occurred during $q$. 

    Let $\T_a, \T_b$ be Alice's and Bob's respective variables at the beginning of iteration~$k$, and denote by $\T_a', \T_b'$ their respective variables at the end of iteration $k'$. 
    If $\T_a, \T_b$ are \good transcripts, then there exists a matching execution of $\pi'$ for $q$,
    with $c$ edit-corruptions, such that $2f - c \geq 0$ if $q$ is the single segment within the frame, or $2f - c \geq -1$ otherwise. It holds that $\T'_a, \T'_b$ are \good.
\end{lemma}

\begin{proof}
    The proof is divided according to the segment types, as in the proof of Lemma \ref{lem:reduction:not-last}.
    The progress each party makes during a segment is determined only by the segment type (regardless to its position in the frame, etc.). Therefore the analysis of what Alice and Bob send and receive within a segment~$q$ is identical to Lemma~\ref{lem:reduction:not-last}. Hence, 
    the matching execution of~$\pi'$ for each case, along with the noise pattern that derives it, is identical to the execution of~$\pi'$ in the corresponding case of Lemma~\ref{lem:reduction:not-last}. Further, To avoid unnecessary clutter, we will not repeat these parts in this proof.
        
    Despite the above similarity, the noise pattern (in~$\Pi$) that induces the parties to take the same actions as in Lemma~\ref{lem:reduction:not-last} may differ in this case, which we analyze below.  
    
    \begin{enumerate}[topsep=0pt,parsep=0pt,label={type (\arabic*)},left=0pt]
        \item 
        The execution of $\pi'$, depicted in Lemma~\ref{lem:reduction:not-last},  has $c=1$ for this case. 
        Then, $2f-c \geq -1$ holds trivially. 
        Note that by Lemma~\ref{lem:segment:props}, there should be more segments in the frame except from~$q$.

        \item 
        Recall that the execution of $\pi'$, depicted in Lemma~\ref{lem:reduction:not-last}, satisfies $0 \le c\le 2$ counting the events $\sigma \ne \sigma'$ and $\rho \ne \rho'$.  
        
        By Lemma~\ref{lem:segment:props}, $q$ is the single segment in the frame, and there is a single iteration in the frame. During this iteration, $\sigma$ and $\rho$ are the messages sent by Alice and Bob respectively (lines~\ref{line:alice:send}, \ref{line:bob:send}), while $\sigma'$ and $\rho'$ are the corresponding received messages (lines~\ref{line:bob:receive}, \ref{line:alice:receive}).
        If $\sigma \ne \sigma'$, a flip must have occur during the transmission of~$\sigma$ (i.e., the $m$-part), and similarly for $\rho,\rho'$. Hence $f\ge c$, and $2f-c \geq 0$.

        \item Following the execution of~$\pi'$ in Lemma~\ref{lem:reduction:not-last}, $c=|P|-1+\delta$. Considering $\Pi$, Lemma~\ref{lem:segment:flips} suggests that during~$q$ at least $(|P|-1)/2$ $r$-bits are flipped.
        Additionally, in a similar manner to type (2) above, at least one flip is required for corrupting $\sigma_{k_1}$ into $\sigma'_{k_1}$ in iteration $k_1$, and at least one flip is required for corrupting $\rho_{k'}$ into $\rho'_{k'}$ in iteration $k'$, implying at least $\delta$ additional flips of $m$-bits.
        We conclude that $f \geq (|P|-1)/2+\delta$, leading to the desired $2f-c \geq 0$.

        \item Following the execution of~$\pi'$ in Lemma~\ref{lem:reduction:not-last}, $c=|P|+\delta$. By Lemma~\ref{lem:segment:flips}, at least $|P|/2$ $r$-bits are flipped during $q$.
        Using a similar argument to case~(3) above, there are least $\delta$ additional flips of $m$-bits. We conclude that $f \geq |P|/2+\delta$, and the condition  $2f-c \geq 0$ holds.
    \end{enumerate}
\end{proof}

\subsubsection{Conclusion}
\label{apdx:thm:main:conclusion}
Now, we can complete the proof of Theorem~\ref{thm:frame}: we split the frame into segments, and use Lemmas \ref{lem:reduction:not-last} and \ref{lem:reduction:last} to reduce each segment into an execution of $\pi'$. Then we wrap all these executions up and create the matching execution for the frame.

\begin{proof}[Proof of Theorem~\ref{thm:frame}]
    Given a frame $s_i, \ldots, s_j$, it can be divided into $R\ge1$ disjoint segments $\{q_1, \ldots, q_{R}\}$ by Lemmas~\ref{lem:frame:segments-cover} and~\ref{lem:segment:not-cross}, in the following manner. Each iteration with progress in the frame belongs to a segment, and all the $R$ segments are contained within the frame. By Lemma \ref{lem:frame:segments-cover}, $q_R$ ends at $j'$, and $q_1$ starts at $i$.

    Consider the case where $R=1$. In this scenario, $q_1=[i,j']$ is the sole segment within the frame. Lemma~\ref{lem:reduction:last} states that there is a matching execution of~$\pi'$ for $q_1$ (and hence, for the entire frame), where at most $2f$ edit-corruptions occurred. $\T_a, \T_b$ are \good by assumption, and the lemma asserts that  $\T'_a, \T'_b$ are \good as well.

    Now consider the case where $R>1$. We construct an execution of $\pi'$ that consists of the concatenated executions matching $q_1,\ldots,q_R$ by activating the above lemmas serially on the segments. In particular, 
    Lemma~\ref{lem:reduction:not-last} suggests that $q_1$ has a matching execution of~$\pi'$ that ends with \good transcripts. Between $q_1$ and $q_2$ there may be consecutive iterations with no progress, so $q_2$ starts with the same $\T_a, \T_b$ as at the end of $q_1$. Thus the lemma can be activated again on $q_2$. This continues sequentially for all $q_1, \ldots, q_{R-1}$. 
    Now, we use Lemma~\ref{lem:reduction:last} to construct an execution of~$\pi'$ matching to~$q_{R}$; again recall that the transcripts at the end of the execution matching $q_{R-1}$ were \good. This implies a matching execution for the entire frame, and that $\T'_a,\T'_b$ at the end of the frame are also \good. 
    The noise pattern of the concatenated execution is the concatenation of the noise patterns in each segment (indeed, this is valid since the transcripts are \good, that is, Alice is the next one to speak at the beginning of each execution).

    Now we analyze the noise in the concatenated execution of~$\pi'$ that matches the entire frame. 
    For $l\in\{1,\ldots, R\}$, denote $f_l$ as the number of flips that occur in~$q_l$, and let $c_l$ be the number of edit-corruptions in the matching execution of~$q_l$. By Lemma~\ref{lem:reduction:not-last}, $c_l  < 2f_l$ for all $l \in \{1, \ldots, R-1\}$. In addition, $c_R \le 2f_R+1$, by Lemma~\ref{lem:reduction:last}. Since $R>1$, it follows that $c=\sum_{l=1...R} c_l < \sum_{l=1..R-1} 2f_l + (2f_R+1) \le 2f+1$, hence the theorem holds.
\end{proof}

\subsection{Completing Theorem \ref{thm:main}}
\label{apdx:thm:main}

In this section we complete the analysis and prove Theorem~\ref{thm:main}. 
The proof consists of three stages: the first is the reduction of a complete execution of $\Pi$ to a matching execution of $\pi'$; the second stage is the connection between the flips count in the execution of~$\Pi$ and the edit-corruptions in the execution of~$\pi'$; and the third is the connection between the probabilities $\{p_i\}_{i\in \mathbb{N}}$ and the flips count in~$\Pi$.

The following lemma states the reduction of a complete execution of $\Pi$ to a matching execution of~$\pi'$.

\begin{lemma}
\label{lem:full-execution}
    Consider an execution of $\Pi$ in which $f$ bit flips occur. Denote by $j$ the last iteration where Alice is active. Then, there exists a matching execution for the consecutive iterations $[1, j]$, in which at most $2f$ edit-corruptions occur.
\end{lemma}

In order to prove Lemma \ref{lem:full-execution}, we split an execution of $\Pi$ into frames and use Theorem \ref{thm:frame} to reduce each frame to an execution of $\pi'$. Nevertheless, we need to consider one final edge case; what happens when the protocol ends. Alice terminates whenever she has simulated enough symbols. However, Bob might be unaware of this event and might not be synchronized with Alice. Technically speaking, the last progress of Alice that causes a termination will signify the end of a segment (the last segment in the execution), but, if Bob is not synchronized, this might not be the end of the sequence, and thus, not the end of a frame. Even worse, once Alice quits, there will never be another insertion by Alice, thus the ``current'' frame will never end. 

To solve this issue, we use the lemmas in section \ref{apdx:segments} until the last frame that has a proper end (by Definition~\ref{def:frame}) and then allow the protocol to run one more iteration as if Alice did not terminate. This added iteration effectively causes the last sequence to be completed and allows us to use the above lemmas for the termination part as well. Of course, this added iteration is used only for analysis purpose.

\begin{lemma}
\label{lem:redundant}
    Consider an execution of $\Pi$ on inputs $x,y$. Let $j$ be the last iteration in which Alice is active (i.e., the iteration where Alice executes Line~\ref{line:alice:output}). 
    Let $j'\le j$ be the last iteration which is the end of a frame. Assume $j' < j$, and denote $i=j'+1$.
    
    Denote by $f$ the number of bit flips that occur during $[i, j]$. 
    Let $\T_a, \T_b$ be Alice's and Bob's respective variables at the beginning of iteration $i$, and denote by $\T_a', \T_b'$ the respective variables at the end of iteration~$j$. 

    If $\T_a, \T_b$ are \good, then there exists 
    a matching execution of $\pi'$
    for the consecutive iterations $[i, j]$, with at most $2f$ edit-corruptions. 
\end{lemma}

\begin{proof}
    First, we show that iteration $j$ cannot be the last iteration of a sequence. The while loop in Algorithm \ref{alg:schemeALICE} ends at the last iteration in which Alice makes progress. Hence, Alice makes progress in iteration~$j$. Thus, if iteration $j$ is the last iteration of a sequence, then by frame definition iteration $j$ is the last iteration of a frame. Thus, $j'=j$, in contradiction. Consequently, $r_a \neq r_b \mod 2$ at the end of~$j$. 

    Since the machinery we have built so far assumes complete frames, 
    we are now going to allow the algorithm~$\Pi$ to run one more iteration (without any noise), for analysis purpose only. 
    At this added iteration~$j+1$, Alice sends Bob the $(|\pi'|+1)$-th symbol of $\pi'$, which we set to be an arbitrary symbol. Alice raises $r_a$ in the beginning of $j+1$, so $r_a = r_b \mod 2$ and Bob receives $r_a$ which does not fit to the expected in Line~\ref{line:bob:if}. Thus, Bob sends Alice the message he sent in iteration $j$, and does not change $r_b$. Alice receives $r_b$ which fits to the expected in Line~\ref{line:alice:if}, inserts the symbol she received from Bob, and we have completed a good sequence. 

    Thus, iteration $j+1$ is the last iteration of a frame, which starts at some $k \geq i$. By frame definition, there is no progress during iterations $[i, k-1]$ (if $k>i$). Hence, we have a complete frame, and we can use Lemma~\ref{lem:frame:segments-cover} to divide the frame $[k, j+1]$ into disjoint segments $q_1, \ldots, q_l$, which cover every iteration with progress in $[k, j+1]$. Note that $q_l = [j+1, j+1]$, since Alice makes progress in~$j$. 
    
    By Lemma~\ref{lem:reduction:not-last}, there exists a matching execution for every segment in $\{q_1, \ldots, q_{l-1}\}$; note, \emph{without~$q_l$}. 
    For $t\in[l-1]$, we denote the flips count in $q_{t}$ by~$f_{t}$, and use Lemma~\ref{lem:reduction:not-last} to obtain a matching execution for~$q_{t}$ with noise pattern that contains $c_{t} < 2f_{t}$ edit-corruptions. We may concatenate all these matching executions, since each segment starts at the final state of the previous segment. 
    The resulting execution is a matching execution for $[i, j]$, with less than $2f$ edit-corruptions.
\end{proof}

The proof of \Cref{lem:full-execution} is then immediate.

\begin{proof}[Proof of Lemma \ref{lem:full-execution}]
    Let $j'\le j$ be the last iteration which is the end of a frame. By \Cref{obs:frame-complete}, each iteration with progress before $j'$ belongs to a frame. By Theorem~\ref{thm:frame}, every frame with $f'$ flips has a matching execution of $\pi'$ with at most $2f'$ edit-corruptions. If $j' < j$, by Lemma~\ref{lem:redundant} there is a matching execution for $[j'+1, j]$. Let $f''$ represent the number of flips occurring during iterations $[j'+1, j]$; then, there are at most $2f''$ edit-corruptions during the matching execution. We use the technique of concatenating these executions of $\pi'$, where each series of iterations in $\Pi$ (either a frame or $[j'+1, j]$) starts at the final state of the previous one. We receive a matching execution for $[1, j]$, with at most $2f$ edit-corruptions.
\end{proof}

\begin{lemma}
\label{lem:success-by-hsv}
    Consider an execution of $\Pi$ over UPEF model where $f$ bit flips occur, and suppose that $f < N'/90$. Then, the output of $\Pi$ at the end of the execution equals to the correct output of $\pi$ with probability of at least $1-2^{-\Theta(N)}$.
\end{lemma}

\begin{proof}
    Recall that $\pi'$ is intended for the insertion-deletion model (Section~\ref{sec:id-model}), whose length is~$N'\in \Theta(N)$. By Theorem \ref{thm:HSV18}, $\pi'$ can withstand a fraction $\delta=1/45$ of insertions and deletions, and outputs the correct output of $\pi$ with probability $1-2^{-\Theta(N)}$. 
    
    Since $f < N'/90$, it follows from Lemma~\ref{lem:full-execution} that there is a matching execution of~$\pi'$ for the execution of~$\Pi$, in which at most $N'/45$ edit-corruptions occur. We argue that during this execution of $\pi'$ Alice uses the insertions-deletions channel $N'$~times. Alice terminates in $\Pi$ when her variable is some $\T_a$ whose length is~$N'$ (by Line \ref{line:alice:while}). In accordance with the definition of matching execution, Alice's variable at the end of the execution of~$\pi'$ is identical to~$\T_a$. 
    
    Thus, we can guarantee that the output of $\Pi$ in the execution equals to the correct output of $\pi$ with probability~$1-2^{-\Theta(N)}$.
\end{proof}

\begin{lemma}
\label{lem:flips-prob}
    Consider an execution of $\Pi$ over UPEF with probabilities that satisfy $\sum_{i=1}^{\infty}{p_i} < S$.
    Then, for any $\eps>0$,
    the number of flips~$f$ observed during $\Pi$ satisfies $\Pr(f > S(1+\eps)) < 2^{-\Omega_\eps(S)}$.
\end{lemma}

\begin{proof}
    Denote the noise pattern in the execution of $\Pi$ by~$E$. Then, $f = \sum_{i \in E}{f_i}$, with $f_i\sim \text{Ber}(p_i)$ an independent Bernoulli random variable that indicates whether the $i$-th transmission is flipped or erased. 
    To ease the analysis, we impose the constraint that $f^{\infty}\ge f$, where $f^{\infty}$ is defined as the number of flips the execution ``suffers'' assuming $E$ dictates a corruption at each and every iteration.

    By linearity of expectations,
    \[
    \E[f^\infty] = \sum_{i=1}^{\infty}{p_i} < S \text{.}
    \]
   
    Let $\eps > 0$ be fixed. 
    In order to bound $\Pr\big(f > S(1+\eps)\big)$, we set $R = \lceil{S - \sum_{i \in E}{p_i}}\rceil$ and define the independent random variables $y_j \sim \text{Ber}\left(\frac{1}{R}\left(S - \sum_{i \in E}{p_i}\right)\right)$, for all $j \in (R,2R]$. Then, we define the random variable $y = f + \sum_{i=R+1}^{2R}{y_i}$. By linearity of expectations, $\E[y] =\E[f] + R\cdot \E[y_i] = S$. Note that $y \geq f$ with probability 1, so by \cite{LP17} (Lemma 22.5) $y$ stochastically dominates $f$. Thus, for any $z \in \mathbb{R}$, $\Pr\left(y > z\right) \geq \Pr\left(f > z\right)$.

    Since $y$ is a summation of independent Bernoulli random variables, we can use Chernoff's bound in order to bound $\Pr\left(y > S(1+\eps)\right)$, which in turn bounds $\Pr\left(f > S(1+\eps)\right)$:
    \begin{align*}
    \Pr\left(f > S(1+\eps)\right) 	&= \Pr\left(\sum_{i \in E}{f_i} > S(1+\eps)\right) \\
                                         &\leq \Pr\left(\sum_{i \in E}{f_i} + \sum_{i=R+1}^{2R}{y_i} > S(1+\eps)\right) \\
                                         &= \text{Pr}\left(y > \E[y] \times (1+\eps)\right) \\
                                         &\leq \exp\left\{-\E[y] \times \eps^2 / 3\right\} = \exp\left\{-S \times \eps^2 / 3\right\}\\
                                         &= 2^{-\Theta_\eps(S)}
    \end{align*}
    The second transition follows from the stochastic domination of $y$ over $f$, and the fourth transition follows from Chernoff's inequality (Theorem 4.4(2) in \cite{MU17}), assuming $\eps\le 1$. If $\eps>1$, the probability is clearly bounded by $\Pr\left(y > 2\E[y]\right)$, which leads to the same asymptotic bound.
\end{proof}

Finally, we can prove the correctness part of Theorem \ref{thm:main}.

\begin{proof}[Proof of Theorem~\ref{thm:main} (correctness only)]
Given an alternating binary protocol~$\pi$ of length $|\pi|=N$, let $\Pi$ be the protocol obtained by executing Algorithms~\ref{alg:schemeALICE} and~\ref{alg:schemeBOB} on $\pi$.
If the number of bit flips $f$ during an execution of $\Pi$ satisfies $f < N'/90$, then by Lemma~\ref{lem:success-by-hsv} $\Pi$ outputs the correct output of $\pi$ with probability $1-2^{-\Theta(N)}$.

Recall that under the UPEF model, if the $i$-th transmission is corrupted, then it is flipped with probability~$p_i$ (or erased otherwise), where the probability $p_i$ is parameterized by the model.
As stated in Theorem~\ref{thm:main}, we set $p_i=\min\left\{\frac{CN}{i^2}, \frac{1}{2}\right\}$. We proceed to find the desired value of $C$; recall that $N' \ge N$ due to Theorem \ref{thm:HSV18}. Thus,

 \[
\sum_{i=1}^{\infty}{p_i} \le \sqrt{\frac{CN}{2}} + \sum_{i=\lceil{\sqrt{2CN}}\rceil}^{\infty}{\frac{CN}{i^2}} \le \left(\frac{1}{\sqrt{2}} + \frac{\pi^2}{6}\right)CN \le \left(\frac{1}{\sqrt{2}} + \frac{\pi^2}{6}\right)CN' \text{.}
\]

By setting $C = \frac{1}{90 \times 1.1 \times 3}$ we receive $\sum_{i=1}^{\infty}{p_i} < \frac{N'}{90 \times 1.1}$. Then, we can use Lemma~\ref{lem:flips-prob} with $\eps = 0.1$ to conclude that $\Pr\left(f > \frac{N'}{90}\right) < 2^{-\Omega(N)}$ (recall that $N' \in \Theta(N)$ by Theorem \ref{thm:HSV18}). Thus, by Lemma~\ref{lem:success-by-hsv}, the protocol $\Pi$ outputs the correct output of $\pi$ with probability of at least $1-2^{-\Omega(N)}$.
\end{proof}

\subsubsection{Communication Complexity Analysis}
\label{apdx:complexity}
We now proceed to complete the proof of Theorem~\ref{thm:main} by analyzing its communication complexity.
In this section we assume a successful execution of~$\Pi$, which happens, as proven above, with very high probability.
The initial step is to analyze the maximum delay that corrupted transmissions may cause to the simulation.
\begin{lemma}
\label{lem:no-progress:sequence}
Let $s=[i,j]$ be a sequence and let $[i, k]$ with $i \le k \le j$ be its prefix.
Assume that there are $t$ corruptions in $[i, k]$. Then, the number of iterations in $[i,k]$ with no progress is at most~$t$.
\end{lemma}
\begin{proof}
First, assume that $i=j$, so $k=j$. In the case where $t=0$, it follows from Lemma~\ref{lem:sequence:all} that both parties make progress during the single iteration of $s$. If $t>0$, the lemma trivially holds.

Otherwise, $i<j$. As previously established in Lemma \ref{lem:sequence:all}, there is progress during both iterations $i$ and $j$. However, middle iterations may either have progress or not.
Furthermore, it follows from Lemma~\ref{lem:sequence:all}
that each middle iteration must contain at least one corruption. This is also true of iteration $i$, although iteration $j$ does not require this.

We proceed by dividing into cases.
If $k < j$, then each iteration during the interval $[i, k]$ has at least one corruption, and $t\geq k-i+1$. Since there are at most $k-i$ iterations with no-progress during $[i,k]$, the claim holds. 
Otherwise, $k=j$. Consequently, there is at least one corruption in each iteration during $[i,k=j]$, with the exception of $j$. Hence, $t \geq j-i$.
It can be demonstrated that there are at most $j-i-1$ iterations without progress during the interval $[i,k=j]$. This establishes the result for this case.
\end{proof}

The above lemma about a single sequence can now be extended to the entire execution, until the iteration where Alice quits. Note that Alice might terminate before completing the last sequence of the execution, but the lemma holds for any prefix of a sequence and thus also for the last (partial) sequence. In this case, similar to Lemma~\ref{lem:redundant}, we continue the execution of $\Pi$ by one iteration, for analysis purpose.
\begin{lemma}
\label{lem:no-progress:limit}
    Assume the noise pattern satisfy $|E|=T$ in some execution, and denote $j$ as the last iteration in the execution where Alice is active.
    Then, at most $T$ iterations in~$[1, j]$ don't have progress.
\end{lemma}

\begin{proof}
Let $j$ be the last iteration in~$\Pi$ where Alice is active, and let $j'\le j$ be the last iteration which is the end of a sequence. It follows from Observation \ref{obs:sequence-cover} that each iteration during the interval $[1,j']$ is included within a sequence. In particular, this includes the iterations with no progress. Let $s=[i, i']$ be a sequence ($i' \leq j'$) with $t_i$ corruptions. By Lemma~\ref{lem:no-progress:sequence}, at most $t_i$ iterations during $s$ have no progress. Let $t \le T$ be the number of corruptions during $[1, j']$; then, there are at most $t$ iterations with no progress during $[1, j']$.

Now, assume $j' < j$ (if $j'=j$ the proof is complete). Since $j$ is not the end of a sequence, $r_a \neq r_b \mod 2$ holds at the end of iteration~$j$. We aim to apply Lemma~\ref{lem:no-progress:sequence} to $[j', j]$. To do so, we continue the execution of $\Pi$ by an additional iteration, without any noise,  similar to the proof of Lemma~\ref{lem:redundant}. As shown there, $r_a = r_b \mod 2$ at the end of $j+1$, so $[j', j+1]$ is a sequence. We apply Lemma~\ref{lem:no-progress:sequence} to the prefix $[j', j]$, and conclude that during $[1, j]$ there are at most $T$ iterations with no progress.
\end{proof}

Next we analyze the maximal number of iterations \emph{with progress} throughout the entire execution.
\begin{lemma}
\label{lem:progress:limit}
    Assume $|\pi'|=N'$. Let~$j$ be the last iteration (in~$\Pi$) where Alice is active, and let~$f$ be the number of bit-flips during~$[1, j]$.
    Then, there are at most $N'+2f$ iterations with progress during $[1,j]$.
\end{lemma}
\begin{proof}
    By Line~\ref{line:alice:while} of Algorithm~\ref{alg:schemeALICE}, Alice leaves when $r_a = N'/2$. Note that the value of $r_a$ cannot decrease during any iteration. Furthermore, Alice makes progress in some iteration $k$ if and only if $r_a$ is increased by 1 during $k$. Thus, there are at most $N'/2$ iterations in which Alice makes progress. 

    We proceed to count the iterations which end with an insertion of Bob. As demonstrated in Lemma \ref{lem:redundant}, if $j$ is not the last iteration of a sequence, we may continue the execution of $\Pi$ by an additional iteration $j+1$ without any noise, such that $j+1$ is the last iteration of a sequence. By Observation \ref{obs:sequence-cover}, all the iterations $[1,j]$ belong to a sequence. Then, it follows from Lemma~\ref{lem:sequence:all}(2)(a),(c) that Bob may insert only at the first or last iteration of a sequence with multiple iterations. Since during each of these sequences there are exactly two insertions, we may divide these sequences into two types: (1) where Alice inserts at the last iteration (so Bob inserts at the first one), and (2) where Bob inserts at the last iteration. 
    
    Sequence of type (1) contains one insertion of Bob, and the number of such sequences is bounded by the number of iterations where Alice makes progress, i.e., $N'/2$. Sequence of type (2) contains one or two insertions of Bob, but by Lemma~\ref{lem:sequence:all}(2)(c), at least one flip occurs during the last iteration of the sequence. Thus, the number of insertions made by Bob during sequences of type (2) is bounded by~$2f$. Therefore, there are at most $N'/2+2f$ insertions of Bob.

    To conclude, there are at most $N'/2+N'/2+2f = N'+2f$ iterations with progress during~$[1, j]$.
\end{proof}

We can finally complete the proof of the complexity and efficiency part of Theorem~\ref{thm:main}.

\begin{proof}[Proof of Theorem~\ref{thm:main} (efficiency and complexity only)]
Note that $\pi'$ is a randomized efficient protocol due to Theorem~\ref{thm:HSV18}; it is easy to verify that $\Pi$ is also efficient. 
Further, the length of $\pi'$ is $N'\in O(N)$ symbols, such that each symbol is of constant size. 

Consider an execution of Algorithms~\ref{alg:schemeALICE} and~\ref{alg:schemeBOB} assuming some noise pattern~$E$ with~$|E|=T$.
Split the iterations of the execution into ones with no-progress and ones with progress. Lemma~\ref{lem:no-progress:limit} bounds the former by~$T$, and Lemma~\ref{lem:progress:limit} bounds the latter by~$O(N+T)$, totaling in $O(N+T)$ iterations throughout the execution of~$\Pi$. By Definition~\ref{def:iteration} and the fact that $\pi'$ has a constant-size alphabet, each iteration of~$\Pi$ consists of $O(1)$ transmission rounds. Thus, $|\Pi| \in O(N+T)$.
\end{proof}

\section{A UF coding scheme based on the concatenated UPEF scheme}
\label{apdx:uf}
In this section we prove the following Theorem.
\begin{theorem}\label{thm:main:concat:realModel}
Given any two-party binary interactive protocol $\pi$ of length $N$, 
there exists an efficient randomized protocol $\Pi$ of length 
$O((N+T)\log(N+T))$
that simulates $\pi$ 
with probability~$1-2^{-\Omega(N)}$
over a binary channel in the presence of an arbitrary and a priori unknown number~$T$ of corruptions. 
The noise is assumed to be independent of the parties' randomness.
\end{theorem}

This objective is achieved by first applying Algorithms
\ref{alg:schemeALICE} and~\ref{alg:schemeBOB} on~$\pi$, and then compiling the result protocol $\Pi$ into a new protocol, $\Pi'$ (Definition~\ref{def:coded-uf}), by coding each transmission of~$\Pi$ via an Algebraic Manipulation Detection (AMD) code~\cite{Cra08}. Due to this encoding, any corruption (bit flips) that affects a given codeword, is either detected by the AMD code (and the outcome can be thus considered as an erasure), or decoded incorrectly (which results as bit flips to the underlying~$\Pi$). 
However, the latter happens with a small probability that we can determine by setting the parameters of the AMD code. 
We use this similarity to argue that the encoded UF protocol~$\Pi'$ succeeds, since it effectively simulates $\Pi$ over the UPEF model. 

We start with the following theorem, which states the properties of the AMD code we are using.
\begin{theorem}[\cite{Cra08}]
\label{thm:AMD}
Fix some probability~$p$ and denote $k = \lceil{\log{p^{-1}}}\rceil + 1$. There exists an efficient probabilistic encoding map $\mathcal{E}: \{0, 1\}^k \rightarrow \{0, 1\}^{3k}$, together with an efficient deterministic decoding function $\mathcal{D}: \{0, 1\}^{3k} \rightarrow \{0, 1\}^k \cup \{\perp\}$, such that for any $m \in \{0, 1\}^k$, $\mathcal{D}(\mathcal{E}(m)) = m$. In addition, for any $m \in \{0, 1\}^k$ and $\Delta \in \{0, 1\}^{3k}$, $\Pr\big(\mathcal{D}(\mathcal{E}(m) + \Delta) \notin \{m, \perp\}\big) \leq p$.
\end{theorem}

Since we transmit bits in $\Pi$, we state the following corollary for coding single bits rather than bit strings.

\begin{corollary}\label{cor:AMD}
    Fix some probability~$p$ and denote $k = \lceil{\log{p^{-1}}}\rceil + 1$. 
    There exists an efficient probabilistic encoding map $\mathcal{E}: \{0, 1\} \rightarrow \{0, 1\}^{3k}$ and an efficient deterministic decoding function $\mathcal{D}: \{0, 1\}^{3k} \rightarrow \{0, 1, \perp\}$, such that $\mathcal{D}(\mathcal{E}(b)) = b$ for any $b \in \{0, 1\}$. In addition, for any $b \in \{0, 1\}$ and $\Delta \in \{0, 1\}^{3k}$, $\Pr\big(\mathcal{D}(\mathcal{E}(b) + \Delta) \notin \{b, \perp\}\big) \leq p$.
\end{corollary}

\begin{proof}
    Given $b\in\{0,1\}$, 
    select uniformly random value $x\in \{0,1\}^{k-1}$, and
    apply the encoding $\mathcal{E}',\mathcal{D}'$ defined by Theorem~\ref{thm:AMD} on~$x\circ b$.
    Formally, set $\mathcal{E}(b) = \mathcal{E}'(x\circ b)$ and for any $s \in \{0, 1\}^{3k}$, 
    set $$\mathcal{D}(s)=\begin{cases} \mathcal{D}'(s)_k & \mathcal{D}'(s)\ne \perp \\ \perp &\text{otherwise}\end{cases}$$
    where $\mathcal{D}'(s)_k$ is the $k$-th bit of~$\mathcal{D}'(s)$.
    It is clear that for any $b \in \{0, 1\}$, $D(\mathcal{E}(b)) = b$ and for any $\Delta \in \{0, 1\}^{3k}$, $\Pr\big(\mathcal{D}(\mathcal{E}(b) + \Delta) \notin \{b, \perp\}\big) \leq p$.
\end{proof}

We proceed to define an algorithm for encoding $\Pi$ into $\Pi'$, using the AMD code.

\begin{definition}[Coded UF]
\label{def:coded-uf}
Let $\Pi$ be a protocol defined over the channel with input alphabet $\{0, 1\}$ and output alphabet~$\{0, 1, \perp\}$, and fix some series of probabilities $\{p_i\}_{i=1}^\infty$. We define $\Pi'$ to be the \emph{UF-coded} protocol (of~$\Pi$),  over the channel with alphabet~$\{0, 1\}$.

For each round $i$ in $\Pi$, let $\mathcal{E}_i:\{0,1\}\to\{0,1\}^{3k_i}$ and $\mathcal{D}_i:\{0,1\}^{3k_i}\to \{0,1,\perp\}$ be the respective encoding and decoding functions given by Corollary~\ref{cor:AMD} when applied to probability~$p_i$.
The protocol~$\Pi'$ follows the computation of $\Pi$ in the following sense. For round $i$ of $\Pi$, let $b_i$ be the message sent in that round. Then, $\Pi'$ simulates that round by $3k_i$ consecutive rounds, where the sender communicates the string $\mathcal{E}_i(b_i)$ over the channel bit-by-bit. The receiver obtains the $3k_i$ bit string~$b'$, computes $\mathcal{D}_i(b')\in\{0,1,\perp\}$, and feeds the decoded symbol to~$\Pi$. That continues until $\Pi$ terminates and gives output.
\end{definition}

With the above definition, we can take any protocol $\Pi$ and encode it into a protocol $\Pi'$. 
We show that this encoding, when applied on $\Pi$ from Theorem~\ref{thm:main}, yields a protocol $\Pi'$ which is correct with high probability.

\begin{lemma}
\label{lem:encoding-reduction}
    Let $\pi$ be some protocol over alphabet $\{0, 1\}$ whose length is $N$, and denote $\Pi$ as defined in Theorem~\ref{thm:main} and $N'$ as defined in~$\Pi$. Further, let $\{p_i\}_{i=1}^\infty$ be some probabilities series, satisfying $\sum_{i=1}^{\infty}{p_i} < \frac{N'}{90 \times 1.1}$. Denote by $\Pi'$ the UF-encoding (by \Cref{def:coded-uf}) of $\Pi$ along with the series $\{p_i\}_{i=1}^\infty$, and let $E_{\text{UF}}$ be a noise pattern in the UF model with some $T > 0$. Then, the probability that the output of $\Pi'$ over $E_{\text{UF}}$ is the correct output of $\Pi$ is at least $1 - 2^{-\Omega(N)}$.
\end{lemma}

\begin{proof}
    Consider some execution of $\Pi'$ with noise pattern $E_{\text{UF}}$. Split the bits accepted on the channel during the execution into words, such that the length of the $i$-th word is $3k_i$. Note that each word represents an encoding of a bit, which may be corrupted by $E_{\text{UF}}$. 
    
    In the event that the encoding of the $i$-th bit is corrupted, the probability that the receiver incorrectly decodes the respective codeword is $q_i$. This probability is bounded by $p_i$, as demonstrated in Corollary \ref{cor:AMD}. We mark by $f$ the count of such cases. 
    
    During the execution of $\Pi'$, the parties use $\Pi$ for the computation of the messages they send. Thus, we may compare the execution of $\Pi'$ to an execution of $\Pi$ over a channel with input alphabet $\{0, 1\}$ and output alphabet $\{0, 1, \perp\}$. The $i$-th bit in that execution of $\Pi$ suffers a flip when the $i$-th bit is decoded incorrectly in the execution of $\Pi'$. Otherwise, the $i$-th bit is either communicated intact or erased, depending on the relevant decoding in $\Pi'$. In this context, $f$ represents the total number of flips. Therefore, it is possible to bound the number of flips that occur during the execution of $\Pi$ by applying the same analysis used in Lemma \ref{lem:flips-prob}.
    By taking 
    $\eps = 0.1$, it follows from $\sum_{i=1}^{\infty}{q_i} \le \sum_{i=1}^{\infty}{p_i} < \frac{N'}{90 \times 1.1}$ that $\Pr(f > N'/90) < 2^{-\Omega(N)}$ (recall that $N' \in \Theta(N)$). 
    As the parties' outputs in the execution of $\Pi'$ are defined as the output of the execution of $\Pi$, it can be shown in a similar manner to Lemma~\ref{lem:success-by-hsv} that the probability that $\Pi'$ outputs the correct output of $\Pi$ is $1 - 2^{-\Omega(N)}$.
\end{proof}

We are ready to complete the proof of Theorem \ref{thm:main:concat:realModel}.
\begin{proof}[Proof of theorem~\ref{thm:main:concat:realModel}]
    Given a protocol $\pi$, apply Theorem~\ref{thm:main} and receive the protocol $\Pi$. Then, set $\{p_i\}_{i=1}^{\infty}$ as defined in Theorem~\ref{thm:main}, i.e., $p_i=\min\left\{\frac{CN}{i^2}, \frac{1}{2}\right\}$, when $C$ is a constant defined in the proof of Theorem \ref{thm:main}. Denote $\Pi'$ as the UF-encoding of $\Pi$, with probabilities $\{p_i\}_{i\in \mathbb{N}}$. We argue that $\Pi'$ is an efficient protocol, which simulates $\pi$ over UF channel with probability of at least $1-2^{-\Omega(N)}$. Assuming there are $T$ corruptions in some execution of $\Pi'$, we argue that $|\Pi'| \in O\left((N+T)\log{(N+T)}\right)$.

    Fix a noise pattern $E_{\text{UF}}$ over UF. As shown in the proof of Theorem~\ref{thm:main}, $\sum_{i=1}^{\infty}{p_i} < \frac{N'}{90 \times 1.1}$, where $N'$ is defined in Algorithms~\ref{alg:schemeALICE} and~\ref{alg:schemeBOB}. Then, by Lemma~\ref{lem:encoding-reduction}, the probability that the output of $\Pi'$ is the correct output of $\Pi$ is at least $1 - 2^{-\Omega(N)}$. Since the correct outputs of $\Pi$ and $\pi$ are identical, $\Pi'$ simulates $\pi$ with probability of at least $1-2^{-\Omega(N)}$ and the correctness part of the proof is complete.

    We proceed to the computation of $|\Pi'|$. By Theorem~\ref{thm:main}, $|\Pi| \in O(N+T)$. Due to the definition of $\Pi'$, its communication complexity is the total length of the encodings of the messages in $\Pi$. Thus,
    
    \allowdisplaybreaks[2]
    \begin{align*}
    |\Pi'| 	&= \sum_{i=1}^{|\Pi|}{3(\lceil{\log{p_i^{-1}}}\rceil+1)} 
               \le  6|\Pi| + 3\sum_{i=1}^{|\Pi|}{\log{p_i^{-1}}} \\
            &\le 6|\Pi| + 3\sum_{i=1}^{\lfloor{\sqrt{2CN'}}\rfloor}{1} + 3\sum_{i=\lceil{\sqrt{2CN'}}\rceil}^{|\Pi|}{\log\left(\frac{i^2}{CN'}\right)} \\
            &\leq 6|\Pi| + 3\sqrt{2CN'} + 3\sum_{i=\lceil{\sqrt{2CN'}}\rceil}^{|\Pi|}{\log\left(\frac{|\Pi|^2}{CN'}\right)} \\
            &\leq 6|\Pi| + 3\sqrt{2CN'} + 3\left(|\Pi| - \sqrt{2CN'}\right)\log\left(\frac{|\Pi|^2}{CN'}\right)
    \end{align*}

    Substituting $|\Pi| \in O(N+T)$ and $N' \in O(N)$ (due to Theorem~\ref{thm:HSV18}), 
    \begin{align*}
    |\Pi'| \in O\left(N+T+\sqrt{N} + (N+T-\sqrt{N})\log(N+T)\right) \subseteq O((N+T)\log(N+T))
    \end{align*}

    The efficiency of $\Pi'$ is trivial and follows from the efficiency of $\Pi$ and of the code from Corollary~\ref{cor:AMD}.
\end{proof}

\subsection{Termination}
\label{sec:termination}
The algorithms \ref{alg:schemeALICE} and~\ref{alg:schemeBOB} employ the following special method for termination that leverages the termination mechanism of the UPEF model. Alice terminates whenever she has completed $N'/2$ iterations of challenge-response with Bob. Subsequently, she has conducted a simulation of the entire transcript, and is confident that Bob has correctly simulated these $N'/2$ iterations, based on his correct response in the last iteration.
Bob, however, does not know whether or not Alice has received this last reply. Therefore, he cannot know when it is safe for him to quit. Recall that the UPEF model assumes that once Alice terminates, the channel transmits a special ``silence'' symbol, namely,~`$\square$'. As Bob receives this special symbol, he knows that Alice has quit, and terminates as well. 

In the UF model, this assumption is no longer tenable. Instead, we are guaranteed that once Alice quits, the channel continues to transmit a default symbol, say a~`$0$`. This symbol might still be corrupted by the channel. Termination in this setting is non-trivial and sometimes even impossible, see discussions in~\cite{DHMSY18,GI19}. However, since the noise is oblivious, we can devise methods for termination, similar to the one in~\cite{DHMSY18}. In a nutshell, Alice terminates as before, and Bob terminates whenever he sees ``enough'' zeros - a \emph{termination string}.

There are two caveats to this approach. One of them is the situation where Bob inserts to his transcript symbols received after Alice has terminated, and outputs a wrong output. However, this is not possible, since when Alice terminates Bob has simulated enough symbols to get the correct output; since Bob commits to an output during the matching execution of $\pi'$ (see \cite{HSV18}), the symbols received afterwards do not affect Bob's output.

The second hazard is that noise may corrupt some of Alice's transmissions, to make Bob quit before Alice quits. This could occur, for instance, if Alice's transmissions were fixed and the noise could transform a lengthy sequence of them into the 0 string. Nevertheless, we demonstrate that given the randomized transmissions of Alice and the oblivious nature of the noise, the probability of this event occurring decreases with the length of the simulated protocol. 

We note that if the length of the termination string depends only on~$N$ (but not on~$T$), then we can set $T$ to be sufficiently large so that the overall probability to force Bob to prematurely terminate, tends to one. Consequently, we set the length of the termination string in an adaptive way, namely to $N+4\log i$, given that the current simulated round in~$\Pi$ is~$i$. In such a manner, the length of the termination string increases as a function of the noise so far, which is some estimation of~$T$. Indeed, if there is no noise, the execution progresses fast and ends by round~$N'$, while if $T$ corruptions are experienced during the execution, it ends by round~$O(N+T)$, by Theorem~\ref{thm:main}.

Before we analyze this type of termination, let us portray the setting.
We assume the protocol~$\Pi$ and the series $\{p_i\}_{i\in \mathbb{N}}$ defined in Theorem \ref{thm:main}. 
Recall the values $\{k_i\}_i$, defined in Theorem \ref{thm:AMD} as $k_i = \lceil \log p_i^{-1}\rceil+1$. We set the protocol~$\Pi'$ to be the one described in Definition~\ref{def:coded-uf}, given the above $\Pi$ and $\{p_i\}_{i\in \mathbb{N}}$. 

Similar to the proof of Lemma \ref{lem:encoding-reduction}, 
we are going to analyze $\Pi'$ through rounds of $\Pi$, recalling that 
$\Pi'$ simulates $\Pi$ round-by-round. 
From now on, we use the term \emph{round} only in context of~$\Pi$, and the term \emph{bit} only for message (of single bit) communicated in~$\Pi'$.

\begin{definition}[Bob's Termination Procedure]
\label{def:bob-terminate}
    By definition \ref{def:coded-uf}, $\Pi'$ simulates $\Pi$ such that Alice encodes the $i$-th round into $3k_i$ bits, and sends to Bob. For each round $i$ in $\Pi$ where Alice sends a message, Bob observes the $t_i = \lceil{N+4\log{i}}\rceil$ bits he receives from Alice, starting at the encoding of round $i$. Bob terminates at the first time where these received bits are all zeros. 
\end{definition}

Before we prove the correctness of this termination procedure, we state the following beneficial property pertaining to the code of \Cref{cor:AMD}.

A possible construction for the AMD code in \cite{Cra08} (which is used in \Cref{cor:AMD}) can be seen as the randomized function $\mathcal{E}'(s)= (s,x,f(s,x))$, where $s\in\{0,1\}^k$ is the input, $x\in\{0,1\}^k$ is a uniform value selected by $\mathcal{E}$, and $f(\cdot, \cdot)$ is some polynomial of $s$ and $x$. Then, let $b\in\{0,1\}$ be some input and let $\mathcal{E}$ be the encoding map given by Corollary \ref{cor:AMD} for some $p$ and $k$ (note that $k \geq 2$ since $p<1$). Recall that $\mathcal{E}$ feeds $\mathcal{E'}$ with input $s = x' \circ b$, such that $x'$ is selected uniformly out of $\{0,1\}^{k-1}$. Consequently, for every prefix of $\mathcal{E}(b)$ whose length is $k' \le 3k$, at least $k'/2$ bits of the prefix are distributed uniformly.

\begin{lemma}
    Let $j$ be the round of $\Pi$ in which Alice terminates. Denote as $\Eterm$ the event where Bob terminates before Alice.
    Then, $\Pr(\Eterm) \in 2^{-\Omega(N)}$.
\end{lemma}
\begin{proof}
    Bob terminates before round $j$ if and only if following two conditions are met. First, there exists some $i \le j$, such that Alice is the sender in round~$i$ of~$\Pi$. Second, all the $t_i$ bits received by Bob starting at the encoding of round~$i$ are zeros, and were sent by Alice (and not by the adversary). 
    Thus, we bound $\Pr(\Eterm)$ in the following manner. First, we compute the probability that Bob receives a 0-string at length $t_i$ starting at the encoding of round~$i$, given that Alice has not terminated yet. We then employ a union bound over all $i\le j$ to bound $\Pr(\Eterm)$.

    Let $i \le j$ be a round where Alice sends a message in~$\Pi$. 
    Let $\hat r_i$ be the $t_i$ bit-long string received by Bob starting at the beginning of the encoding that corresponds to round $i$ in~$\Pi$.
    Alice's transmissions, which result in the string $\hat r_i$,  consist of all the encodings of messages in~$\Pi$ sent by Alice between round $i$ and some round $i'$ (such that the encoding of $i'$ ends after sending $t_i$ bits). By Definition \ref{def:coded-uf}, the length of the encoded message transmitted in a round $\rho\in[i,i']$ is $3k_{\rho}$. 
    By the property we stated after \Cref{def:bob-terminate}, at least $3k_{\rho}/2$ bits of an encoding of length $3k_{\rho}$ distribute uniformly. 
    Since the noise is oblivious, the probability that instead of an encoding, Bob receives a 0-string of length $3k_{\rho}$ (including the case where there is no noise) is bounded by $2^{-3k_{\rho}/2}$.
    Note that $\hat r_i$ may end before in a middle bit of an encoding. However, by the same property, the probability of this suffix of $\hat r_i$, assuming its length is some $k'$, to be a 0-string is at most $2^{-k'/2}$.
    
    It follows then,
    that the probability that $\hat r_i$ is a 0-string is at most~$2^{-t_i/2} \le 2^{-(N/2+2\log(i))}$.
    Therefore,
    \[
    \Pr(\Eterm) \le \sum_{i=1}^{j}{2^{-(N/2+2\log(i))}} = 2^{-N/2}\sum_{i=1}^{j}{\frac{1}{i^2}} <  2^{-N/2+1} \in  2^{-\Omega(N)}\text{.}
    \]
\end{proof}

We proceed to find the communication complexity of our simulation, given this termination procedure.
Assume that Alice terminates at round~$j$ and Bob at round $j'>j$. Hence, there exists a round~$l$ in~$\Pi$, in which Alice is the sender and all the $t_l$ bits received by Bob starting from the encoding of~$l$ are zeros.
The worst case (in terms of communication complexity) is that starting from the encoding of round~$j+1$, the adversary corrupts one bit in each $t_i=\lceil{N+4\log{i}}\rceil$ consecutive bits sent to Bob (\emph{termination string} starting at~$i$).
Since the adversary's budget is limited to~$T$, Bob terminates after $T+1$ such termination strings.
Let us now bound their total length.

Recall that $p_{i} = \min\left\{\frac{CN'}{{i}^2}, \frac{1}{2}\right\}$. 
Since $j>N'$ and $C<1$, it follows that $p_{i} = \frac{CN'}{{i}^2}$ and $k_{i} \geq \log(i^2/CN') \geq \log{i}$ for all $i\in [j+1,j']$.
Therefore, a termination string of length $t_i$ contains at most $L=O((N+4\log i)/(3\log i)))=O(N/\log j)$ encodings. Since there are at most $T+1$ termination strings, then there are at most $O(NT/\log{j})$ rounds between $j$ and $j'$ where Alice is the sender. It follows from the structure of $\Pi$ that after every $|\Sigma'|+1$ rounds where Alice is the sender (recall that $\Sigma'$ is the alphabet of $\pi'$, the protocol defined by Theorem \ref{thm:HSV18}) there are exactly $|\Sigma'|+1$ rounds where Bob is the sender. Thus, there are at most $O(NT/\log{j})$ rounds between $j$ and $j'$. Recall that $j\in O(N+T)$ and $j \in \Omega(N)$, hence $j' \le j+O(TL)= O(N+T)+O(NT/\log {N})$.

At the $i$-th round, an encoding of length $3k_i$ is sent over the channel. Note that $k_i \le 2\log{i}+2$ for all $i \in [j+1,j']$. Thus, the number of bits sent between the termination of Alice and the termination of Bob is
\[
\sum_{i=j}^{j'=j+O(TL)}{3k_i} \le \sum_{i=j}^{j'=j+O(TL)}{3(2\log{i}+2)} \leq O(TL)\cdot 3(2\log{j'}+2)\text{.}
\]
It follows that the additional complexity due to termination is,
\begin{align*}
O(TL)\cdot 3(2\log{j'}+2) 
&\le O\left(\frac{NT}{\log{N}} \cdot \log \left(N+T +\frac{NT}{\log{N}}\right)\right)\text{.}
\end{align*}

As one can see, this approach tremendously increases the communication complexity. It can be amended by relaxing the termination condition so that Bob terminates when he sees ``many'' zeros, say, more than 90\% during the last $t_i$ transmissions. By Chernoff, the probability to see that many zeros given that Alice has not terminated and given any fixed noise, is exponentially small. See, e.g., \cite{DHMSY18}.

\end{document}